\documentclass[journal,twoside,web]{ieeecolor}
\usepackage[table,dvipsnames]{xcolor}
\usepackage{tikz}
\usetikzlibrary{positioning}
\usetikzlibrary{arrows.meta, positioning}
\usepackage{generic}
\usepackage{cite}
\usepackage{amsmath,amssymb,amsfonts}
\usepackage{algorithmic}
\usepackage{algorithm,algorithmic}
\usepackage{subcaption}
\usepackage{graphicx}

\newtheorem{theorem}{Theorem}
\newtheorem{rem}{Remark}
\newtheorem{definition}{Definition}
\newtheorem{lemma}{Lemma}

\usepackage{hyperref}
\usepackage{textcomp}

\definecolor{mygreen}{RGB}{49, 173, 103}

\newcommand{\B}[1]{\mathbf{#1}}

\def\BibTeX{{\rm B\kern-.05em{\sc i\kern-.025em b}\kern-.08em
    T\kern-.1667em\lower.7ex\hbox{E}\kern-.125emX}}
\begin{document}
\title{Leader-Follower Density Control of Multi-Agent Systems with Interacting Followers: Feasibility and Convergence Analysis}
\author{Beniamino Di Lorenzo, Gian Carlo Maffettone, and Mario di Bernardo
\thanks{Beniamino Di Lorenzo and Gian Carlo Maffettone contributed equally to this work (Corresponding author: Mario di Bernardo).}
\thanks{Beniamino Di Lorenzo, Gian Carlo Maffettone and Mario di Bernardo are with the Modeling and Engineering Risk and Complexity program of the Scuola Superiore Meridionale, Naples, Italy (e-mails: b.dilorenzo@ssmeridionale.it, giancarlo.maffettone@unina.it, mario.dibernardo@unina.it). }
\thanks{Mario di Bernardo is also with the Department of Electrical Engineering and Information Technology of the University of Naples Federico II, Naples, Italy.}}

\maketitle

\begin{abstract}
We address density control problems for large-scale multi-agent systems in leader-follower settings, where a group of controllable leaders must steer a population of followers toward a desired spatial distribution. Unlike prior work, we explicitly account for follower-follower interactions, capturing realistic behaviors such as flocking and collision avoidance. Within a macroscopic framework based on partial differential equations governing the density dynamics, we derive (i) necessary and sufficient feasibility conditions linking the target distribution to interaction strength, diffusion, and leader mass, and (ii) a feedback control law guaranteeing local stability with an explicit estimate of the basin of attraction. Our analysis reveals sharp feasibility thresholds—phase transitions beyond which no control effort can achieve the desired configuration. Numerical simulations in one- and two-dimensional domains validate the theoretical results at the macroscopic level, and agent-based simulations on finite populations confirm the practical deployability of the proposed framework.
\end{abstract}


\thispagestyle{empty}

%%%%%%%%%%%%%%%%%%%%%%%%%%%%%%%%%%%%%%%%%%%%%%%%%%%%%%%%%%%%%%%%%%%%%%%%%%%%%%%%
\section{Introduction}

Controlling the spatial distribution of large agent populations is a fundamental challenge in multi-agent systems theory, with applications ranging from traffic management \cite{siri2021freeway, papageorgiou2003review} and swarm robotics \cite{sinigaglia2025robust, dorigo2021swarm, brambilla2013swarm} to crowd evacuation \cite{yuan2023multi, helbing2000simulating} and biological systems\cite{massana2022rectification, giusti2026data}. In many such scenarios, including the mitigation of traffic jams with autonomous vehicles \cite{stern2018dissipation} and search-and-rescue missions with robotic rescuers \cite{yuan2023multi, pierson2017controlling}, control authority is delegated to a subset of leader agents tasked with steering the collective behavior of a larger follower population.

This challenge may be formulated within the framework of leader–follower density control, where the objective is to design the leaders' behavior so as to drive the followers towards a prescribed spatial distribution \cite{almi2023optimal, maffettone2025leader}. In \cite{maffettone2025leader}, we developed a macroscopic approach based on partial differential equations (PDEs) governing the density dynamics, deriving feasibility conditions, that regards the existence of solutions to the control problem, and stabilizing feedback laws. However, a common simplification in this and other works is to assume that follower dynamics are driven solely by leader-follower interactions and stochastic effects \cite{maffettone2025leader, lama2024shepherding, di2025continuification}, thereby neglecting follower–follower (inter-follower) interactions. This assumption overlooks realistic behaviors such as flocking, collision avoidance, and cohesive or evasive maneuvers among followers \cite{strombom2014solving}.

When followers interact with each other, the control problem changes fundamentally. Consider autonomous vehicles guiding human-driven traffic through congestion: human drivers respond not only to the autonomous leaders but also maintain safe distances from other vehicles, form lanes, and exhibit collective phenomena such as phantom traffic jams \cite{stern2018dissipation}. Similarly, in evacuation scenarios, humans naturally maintain personal space and may exhibit herding behavior that can either facilitate or hinder the process \cite{helbing2000simulating}. These inter-follower interactions can dramatically alter the system's controllability, potentially requiring more leaders to achieve a given task, or rendering certain target configurations entirely infeasible.

In this work, we address this gap by incorporating follower–follower interactions into a macroscopic density control framework. Such interactions, whether cohesive forces between vehicles maintaining safe headways, or herding instincts among evacuees, manifest at the population level as nonlinear modifications to the collective velocity field, making a continuum description the natural setting for their analysis. Building on our previous results \cite{maffettone2025leader}, we adopt a continuum approach in which the spatiotemporal evolution of leader and follower densities is governed by coupled PDEs. This macroscopic formulation is part of a broader multi-scale control framework for large agent populations, where continuum-level design can be complemented by macro-micro bridges, such as optimal transport or spatial sampling, to translate density-level commands into individual agent actions \cite{maffettone2025leader, di2025continuification, napolitano2026optimal, maffettone2024mixed, catello2025sparse}. 

A key finding of this work is that follower–follower interactions fundamentally alter feasibility of density control. Specifically, in some scenarios, inter-follower interactions play the role of a disturbance, requiring a larger leader mass to achieve the same target distribution, in comparison with the case with non-interacting followers. In other cases, when the followers' open-loop steady-state displacement is statistically close to the target density, a smaller fraction of leaders is required to achieve the goal. Such existence results have coherent implications in control design: in the absence of inter-follower interactions, we are able to recover feedback laws ensuring global stability, while in the presence of interacting followers, we only recover local stability guarantees.

\subsection{Related works and main contributions}

The control of large-scale leader-follower systems has been tackled across various methodological domains.

Within mean-field optimal control \cite{fornasier2014mean}, extensions to large interacting populations with heterogeneous control authority have been proposed \cite{fornasier2014mean2, bongini2017optimal, almi2023optimal}, with recent advances modeling transient leadership modes \cite{albi2022mean, albi2024kinetic}. While these approaches explicitly account for inter-follower interactions, they formulate optimal control problems constrained by PDE dynamics without deriving closed-form feedback solutions, limiting their practical applicability.

Field-theoretic approaches for multi-population systems have recently been developed in the context of shepherding control \cite{lama2025nonreciprocal}, focusing on decision-making processes in continuum descriptions rather than feedback laws with convergence guarantees. Notably, inter-follower dynamics are not considered. In \cite{bernardi2021leadership, bernardi2019macroscopic}, effective leadership mechanisms are identified such as orientation toward the target, speed modulation, and visibility but no analytical convergence guarantees are provided.

Our work differs from these approaches by combining three elements: (i) explicit modeling of inter-follower interactions within a leader-follower continuum framework, (ii) analytical feasibility conditions with physical interpretation, (iii) closed-form feedback control with provable stability guarantees.

Our main contributions are as follows:
\begin{enumerate}
    \item[(i)] We incorporate follower–follower interactions into a leader–follower density control framework and derive \emph{necessary and sufficient feasibility conditions} linking target densities to interaction strength, diffusion intensity, and leader mass. We show that inter-follower interactions generally make these conditions more restrictive. 
    \item[(ii)] We generalize one of the feedback control laws from \cite{maffettone2025leader} to the interacting-follower case and prove \emph{local asymptotic stability}, together with an explicit estimate of the basin of attraction.
    \item[(iii)] We provide \emph{design guidelines} for practical deployments, quantifying how increased cohesion or repulsion among followers affects the minimum required leader mass and modifies the feasibility region.
\end{enumerate}

While this work focuses on theoretical foundations validated through numerical simulations, it lays the groundwork for future experimental validation on robotic platforms and integration with safety constraints for operation in complex environments. For example, the minimum leader mass derived from our feasibility analysis translates, in finite populations, to a minimum number of leaders, a constraint analogous to herdability conditions studied in shepherding problems \cite{lama2024shepherding}, where one asks what is the minimum number of herders (leaders) required to confine or drive a given number of  targets (followers).

The rest of the paper is organized as follows. In Section~\ref{eq:mod_prob_stat}, we present our mathematical model in a one-dimensional setting and provide the control problem statement. In Section~\ref{sec:feasibility}, we provide a feasibility analysis, returning an existence result for the solution of the control problem. In Section~\ref{sec:control_design}, we show our feedback control strategy ensuring local stability, which we then validate in Section~\ref{sec:numerical_validation}. The mathematical framework is extended to multi-dimensional domains in Section~\ref{sec:ext_hd}, and guidelines for the deployment of our strategy on swarms of finite size are given in Section~\ref{sec:deployment}. Conclusions are presented in Section~\ref{sec:conclusions}.

\section{Model and Problem Statement}\label{eq:mod_prob_stat}

To compactly describe the collective spatial organization of large-scale multi-agent systems, we extend the leader-follower density control framework from \cite{maffettone2025leader} by incorporating inter-follower interactions. Our approach employs two coupled PDEs that capture the spatiotemporal evolution of the population densities. The followers' density dynamics is governed by a convection-diffusion equation, where the diffusive term represents agent-level stochasticity and the convection accounts for both inter-population (follower-leader) and intra-population (follower-follower) interactions. The leaders' density evolution follows a purely convective law driven by a control velocity field $u$. We formulate our problem on the unit circle $\mathcal{S}=[-\pi, \pi]$, yielding
\begin{subequations}\label{eq:themodel}
    \begin{align}
        \rho_t^L(x,t) + \left[\rho^L(x, t) u(x,t)\right]_x &= 0,\label{eq:leaders}\\
        \rho_t^F(x,t) + \left[\rho^F(x,t)\left(v^{FL}(x,t) + v^{FF}(x,t)\right)\right]_x &= D \rho^F_{xx}(x,t),\label{eq:followers}
    \end{align}
\end{subequations}
where $\rho^L,\rho^F : \mathcal{S}\times\mathbb{R}_{\geq 0}\to\mathbb{R}_{\geq 0}$ are the leaders' and followers' densities, respectively, with $x$ and $t$ denoting the spatial and temporal coordinates. The subscripts $t$ and $x$ denote partial derivatives with respect to time and space. Here, $D\in\mathbb{R}_{>0}$ is the diffusion coefficient governing the followers' stochastic behavior, and $u : \mathcal{S}\times\mathbb{R}_{\geq 0}\to\mathbb{R}$ is the control velocity field to be designed. The velocity fields
\begin{subequations}
\begin{align}
    v^{FF}(x,t) &= \int_\mathcal{S} f^{FF}(y\triangleright x) \rho^{F}(y, t)\,\mathrm{d}y = (f^{FF}*\rho^F)(x, t),\label{eq:vff}\\
    v^{FL}(x,t) &= \int_\mathcal{S} f^{FL}(y\triangleright x) \rho^{L}(y, t)\,\mathrm{d}y = (f^{FL}*\rho^L)(x, t),\label{eq:vfl}
\end{align}
\end{subequations}
capture the follower-follower and follower-leader interactions, respectively, where ``$*$" denotes the convolution operator\footnote{Since the problem is formulated on a periodic domain, convolutions should be interpreted as circular convolutions, which naturally preserve periodicity.}, ${y \triangleright x} = (x-y+\pi) \,\mathrm{mod}(2\pi)-\pi$ is the relative position between $x$ and $y$ wrapped on $\mathcal{S}$ and $f^{FF}, f^{FL}:\mathcal{S}\to\mathbb{R}$ are periodic, odd, soft-core interaction kernels \cite{bernoff2011primer, bodnar2005derivation}. We fix follower-leader interactions to be driven by a periodic repulsive interaction kernel, that is
\begin{align}\label{eq:rep_kern}
    f^{FL}(x) = \frac{\mathrm{sgn}(x)}{\mathrm{e}^{2\pi/\ell}-1} \left[\mathrm{e}^{\frac{2\pi-\vert x\vert }{\ell}} - \mathrm{e}^{\frac{\vert x\vert}{\ell}}\right],
\end{align}
where $\ell$ is the characteristic interaction length. 

\begin{rem}
    Our theoretical framework does not assume a specific form for $f^{FF}$. However, in our examples and simulations, we assume follower-follower interactions to be either repulsive or driven by a Morse potential (short-range repulsion and long-range attraction), that is
\begin{align}\label{eq:ff_kernel}
    f^{FF}(x) = \frac{1}{\ell_r} f_r(x) - \frac{\zeta}{\ell_a}f_a(x),
\end{align}
where $\ell_r$ and $\ell_a$ are the characteristic length scale of repulsion and attraction; $\zeta$ is a gain modulating the relative strength of attraction and repulsion, and
\begin{align}
    f_i(x) = \frac{\mathrm{sgn}(x)}{\mathrm{e}^{2\pi/\ell_i}-1} \left[\mathrm{e}^{\frac{2\pi-\vert x\vert }{\ell_i}} - \mathrm{e}^{\frac{\vert x\vert}{\ell_i}}\right],
\end{align}
for $i = a, r$. The case $\zeta =0$ implies repulsive follower-follower interactions\footnote{Note that periodic interaction kernels can be derived by periodization of non-periodic ones -- see \cite{boldini2024stigmergy} for more details.}. We refer to Fig.~\ref{fig:kernels} for a graphical schematization of $f^{FL}$ and $f^{FF}$ (restricting ourself to the case of Morse inter-follower interactions for the sake of a good graphical interpretation).  
\end{rem}
\begin{figure}
    \centering
    \includegraphics[width=0.6\linewidth]{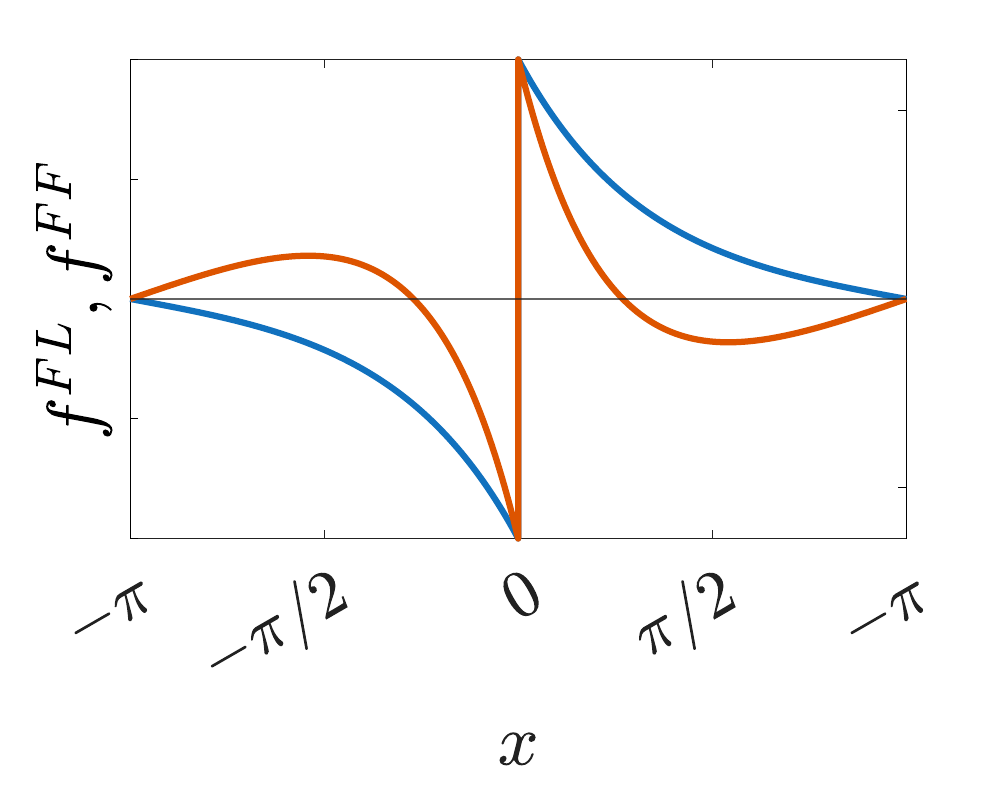}
    \caption{Interaction kernels: in blue the repulsive follower-leader interaction kernel (pure repulsion); in orange, the Morse follower-follower interaction kernel (long-range attraction and short-range repulsion).}
    \label{fig:kernels}
\end{figure}

The model \eqref{eq:themodel} is complemented with initial conditions in $\mathcal{H}^2(\mathcal{S})$\footnote{$\mathcal{H}^2(\mathcal{S})$ is the Sobolev space of functions defined on $\mathcal{S}$ with up to second-order derivative being in  $\mathcal{L}^2(\mathcal{S})$.} 
\begin{align}\label{eq:IC}
    \rho^i(x, 0) &= \rho^i_0(x),\;\forall x \in\mathcal{S},\;\;i = L, F,
\end{align}
and periodic boundary conditions
\begin{align}\label{eq:PBCs}
    \rho^i(-\pi, t) &= \rho^i(\pi, t),\;\forall t\in\mathbb{R}_{\geq 0}, \;\;i=L, F.
\end{align}
We remark that when $u$ is periodic, \eqref{eq:PBCs} ensures conservation of the total mass of the system (noting that $v^{FL}$ and $v^{FF}$ are periodic since they arise from circular convolutions)--see \cite{maffettone2025leader} for further details. We define the mass of leaders and followers as
\begin{align}
    M^i &= \int_\mathcal{S}\rho^i(x, t)\,\mathrm{d}x, \;\;i=L, F,
\end{align}
and, without any loss of generality, we assume 
\begin{align}
    M^L + M^F = 1.
\end{align}
For the sake of interest, we neglect the scenario $M^L=0$ and $M^F=1$ and its converse.

\subsection{Problem Statement}\label{sec:prob_stat}
We seek a periodic control field $u$ in \eqref{eq:leaders} that drives the follower density $\rho^F$ from an initial distribution $\rho^F_0$ to asymptotically converge to a prescribed target distribution
\begin{align}\label{eq:prob_stat}
    \lim_{t\to\infty} \left\Vert \bar{\rho}^F(\cdot) - \rho^F(\cdot, t) \right\Vert_2 = 0,
\end{align}
where $\bar{\rho}^F:\mathcal{S}\to\mathbb{R}_{>0}$ is a time-invariant target density satisfying the mass conservation constraint $\int_\mathcal{S}\bar{\rho}^F(x)\,\mathrm{d}x = M^F$, with $M^F$ denoting the total follower mass, and $\Vert \cdot \Vert_2$ the $\mathcal{L}^2$-norm on $\mathcal{S}$ (we highlight the variable with respect to which the norm is taken with a  ``$\cdot$" in the functions arguments).

\begin{rem}
The macroscopic formulation \eqref{eq:themodel} describes the density dynamics; in Section~\ref{sec:deployment} we focus on the agent-based model underlying such a macroscopic density dynamics.
\end{rem}

\section{Feasibility Analysis}\label{sec:feasibility}
Before designing a control law, we must determine whether a given target distribution is achievable. In this section, we establish {\em necessary and sufficient} feasibility conditions for the density control problem with inter-follower interactions. These conditions explicitly link the target distribution to the interaction strength, diffusion intensity, and available leader mass—providing interpretable criteria for when density control is possible.

\begin{definition}[Feasible Steady-State Solution]\label{def:feasibility}
    The problem described by \eqref{eq:themodel} and \eqref{eq:prob_stat} admits a \emph{feasible steady-state solution} (or analogously, is \emph{feasible}) if, for a given follower mass $M^F \in (0,1)$, there exists a leader density distribution $\bar{\rho}^L : \mathcal{S} \to \mathbb{R}_{\geq 0}$ satisfying the constraints
    \begin{subequations}\label{eq:feasibility_conditions}
        \begin{align}
            \bar{\rho}^L(x) &\geq 0, \quad \forall x \in \mathcal{S}, \label{eq:positivity}\\
            \int_\mathcal{S} \bar{\rho}^L(x)\,\mathrm{d}x &= M^L = 1 - M^F, \label{eq:mass}
        \end{align}
    \end{subequations}
    and such that the target density $\bar{\rho}^F$ constitutes a steady-state solution of the follower dynamics \eqref{eq:followers} under the time-invariant interaction field
    \begin{align}\label{eq:steady_interaction}
        v^{FL}(x, t) = \bar{v}^{FL}(x) = (f^{FL} * \bar{\rho}^L)(x).
    \end{align}
\end{definition}

For the problem to be feasible, we need $\bar{\rho}^F$ to be a solution of \eqref{eq:followers}, that is
\begin{align}\label{eq:followers_ss}
    \left[\bar{\rho}^F(x)(\bar{v}^{FL}(x) + \bar{v}^{FF}(x))\right]_x &= D\bar{\rho}^F_{xx}(x),
\end{align}
with $\bar{v}^{FF} = (f^{FF}*\bar{\rho}^F)$. By a spatial integration of \eqref{eq:followers_ss}, recalling $\bar{\rho}^F\neq 0$ $\forall x\in\mathcal{S}$ (see Section~\ref{sec:prob_stat}), we obtain
\begin{align}
    \bar{v}^{FL}(x) = D\frac{\bar{\rho}^F_x(x)}{\bar{\rho}^F(x)} - \bar{v}^{FF}(x) + \frac{A}{\bar{\rho}^F(x)}, 
\end{align}
with $A$ being the constant of integration. Such a relation holds almost everywhere, being the problem framed in $\mathcal{H}^2(\mathcal{S})$. Taking initial conditions in $\mathcal{C}^2(\mathcal{S})$ would make the relation hold point-wise.

Being $f^{FL}$ odd and recalling the Fubini's theorem for convolutions \cite{royden1988real}, it must hold that 
\begin{align}\label{eq:integral_cond}
    \int_\mathcal{S} \bar{v}^{FL}(x)\,\mathrm{d}x = \int_\mathcal{S} (f^{FL}*\bar{\rho}^L)(x)\,\mathrm{d}x = 0.
\end{align}
Note that the same holds for $\bar{v}^{FF}$, that is $\int_\mathcal{S}\bar{v}^{FF}\,\mathrm{d}x =0$. Hence, we can compute $A$ imposing the integral condition \eqref{eq:integral_cond}, yielding
\begin{align}
    A = -\frac{D\left[\log(\bar{\rho}^F(x))\right]_{-\pi}^\pi}{\int_{\mathcal{S}}1/\bar{\rho}^F(x)\,\mathrm{d}x} = 0,
\end{align}
for the periodicity of $\bar{\rho}^F$. Consequently, the steady-state interaction field must satisfy
\begin{align}\label{eq:vfl_bar}
    \bar{v}^{FL}(x) = D\frac{\bar{\rho}^F_x(x)}{\bar{\rho}^F(x)} - \bar{v}^{FF}(x),
\end{align}
where $\bar{v}^{FF}$ represents the follower-follower interaction field. This velocity field extends the result in \cite{maffettone2025leader} by incorporating the additional term $\bar{v}^{FF}$, which accounts for inter-follower interactions absent in the original formulation.

From \eqref{eq:vfl_bar} and the requirement that $\bar{v}^{FL} = f^{FL}*\bar{\rho}^L$, we can recover the leader density $\bar{\rho}^L$ via deconvolution with the repulsive kernel \eqref{eq:rep_kern} (see \cite{maffettone2025leader, boldini2024stigmergy}). This operation yields
\begin{align}\label{eq:rho_bar_l}
    \bar{\rho}^L(x) = \frac{1}{2}\bar{v}^{FL}_x(x) - \frac{1}{2\ell^2}\int\bar{v}^{FL}(x)\,\mathrm{d}x + B,
\end{align}
where $B$ is an integration constant (see Appendix B of \cite{maffettone2025leader} for the derivation). 

Since deconvolution alone does not guarantee that $\bar{\rho}^L$ satisfies Definition~\ref{def:feasibility}, the feasibility problem reduces to determining whether there exists a value of $B$ in \eqref{eq:rho_bar_l} such that the conditions \eqref{eq:feasibility_conditions} are satisfied.

To simplify the analysis, we express the target density as
\begin{align}
    \bar{\rho}^F(x) = M^F \hat{\rho}^F(x), \quad \text{with} \quad \int_\mathcal{S} \hat{\rho}^F(x)\,\mathrm{d}x = 1,
\end{align}
where $\hat{\rho}^F$ denotes the normalized desired follower density profile. Consequently, equation \eqref{eq:vfl_bar} becomes
\begin{align}\label{eq:vfl_bar2}
    \bar{v}^{FL}(x) = D\frac{\hat{\rho}^F_x(x)}{\hat{\rho}^F(x)} - M^F \hat{v}^{FF}(x),
\end{align}
with $\hat{v}^{FF}(x) = (f^{FF}*\hat{\rho}^F)(x)$.

Before stating our feasibility theorem, we define some useful functions and constants. Specifically,  
\begin{subequations}\label{eq:functions}
    \begin{align}
        g_1(x) &= \left(\frac{\hat{\rho}_x^F(x)}{\hat{\rho}^F(x)}\right)_x = \left[\log(\hat{\rho}^F(x))\right]_{xx}, \label{eq:g1}\\
        g_2(x) &= \log(\hat{\rho}^F(x)), \label{eq:g2} \\
        g_F(x) &= \frac{1}{2\ell^2}\int \hat{v}^{FF}(x)\,\mathrm{d}x -\frac{\hat{v}^{FF}_x(x)}{2}, \label{eq:gF}
    \end{align}
and:
    \begin{align}
        C &= \int_\mathcal{S}\log(\hat{\rho}^F(x))\,\mathrm{d}x, \label{eq:C}\\
        C_F &= \int_\mathcal{S} g_F(x)\,\mathrm{d}x. \label{eq:Cf}
    \end{align}
\end{subequations}
Moreover
\begin{align}\label{eq:h_F}
    h_F(x) = \frac{C_F}{2\pi} - g_F(x).
\end{align}
An interpretation of such quantities follows the next theorem that
provides necessary and sufficient conditions for feasibility. Notably, these conditions reduce to verifying inequalities on the leader mass $M^L$, with the bounds determined by the target distribution, diffusion coefficient and, crucially, the inter-follower interaction kernel through $h_F$.

\begin{theorem}[Feasibility]\label{th:feasibility_theorem}
    Given a normalized target density $\hat{\rho}^F$, the problem described by \eqref{eq:themodel} and \eqref{eq:prob_stat} is feasible according to Definition~\ref{def:feasibility} if and only if the following constraints are fulfilled:
    \begin{subequations}\label{eq:constraints_feasibility}
        \begin{align}
            \widehat{M}^L_1 \leq M^L \leq \widehat{M}^L_2, \quad M^L \in (0,1), \label{eq:constr_mass}\\
            G(\bar{x}) \leq 0, \quad \forall \bar{x}\in\mathcal{S}: H(\bar{x}) = 0, \label{eq:constr_G}
        \end{align}
    \end{subequations}
    where
    \begin{subequations}
        \begin{align}
            \widehat{M}^L_1 &= \max_{x:H(x)>0} \left[h(x)\right], \label{eq:M1hat}\\
            \widehat{M}^L_2 &= \min_{x:H(x)<0} \left[h(x)\right], \label{eq:M2hat}
        \end{align}
    \end{subequations}
    with
    \begin{align}
        h(x) &= \frac{G(x)}{H(x)},\label{eq:h}
    \end{align}
    and
    \begin{subequations}
        \begin{align}
            H(x) &= \frac{1}{2\pi} + h_F(x),\label{eq:HF}\\
            G(x) &= -\frac{D}{2}g_1(x) + \frac{D}{2\ell^2}g_2(x) - \frac{DC}{4\pi\ell^2} + h_F(x),\label{eq:G}
        \end{align}
    \end{subequations}
    where $g_1$, $g_2$, $C$ are defined in \eqref{eq:functions} and $h_F$ in \eqref{eq:h_F}.
\end{theorem}
\begin{proof}
    We start by proving sufficiency ($\Rightarrow$). Given the expression of $\bar v^{FL}$  in \eqref{eq:vfl_bar2}, $\bar \rho^L$ in \eqref{eq:rho_bar_l}
    can be rewritten as
    \begin{align}\label{eq:rho_bar_l2}
        \bar \rho^L(x) = \frac{D}{2}g_1(x) - \frac{D}{2\ell^2}g_2(x)+M^Fg_F(x)+B
    \end{align}
    with $g_1$, $g_2$ and $g_F$ defined in \eqref{eq:functions}. Using \eqref{eq:rho_bar_l2}, we can compute
    \begin{align}
        \int_\mathcal{S} \bar{\rho}^L(x)\,\mathrm{d}x = -\frac{DC}{2\ell^2} + M^F C_F + 2\pi B,
    \end{align}
    where $C$ and $C_F$ are given in \eqref{eq:C} and \eqref{eq:Cf}, and we exploited the periodicity of $\hat{\rho}^F$. To fulfill \eqref{eq:mass} in Definition \ref{def:feasibility}, we choose the arbitrary constant $B$ in \eqref{eq:rho_bar_l2} as
    \begin{align}\label{eq:B}
        B = \frac{1}{2\pi}\left(1-M^F+\frac{DC}{2\ell^2}-M^F C_F \right).
    \end{align}
    Substituting \eqref{eq:B} into \eqref{eq:rho_bar_l2} and imposing non-negativity of $\bar \rho^L$ due to \eqref{eq:positivity} in Definition \ref{def:feasibility}, we obtain
    \begin{multline}\label{eq:rho_bar_L_2}
        \bar{\rho}^L(x) = \frac{D}{2}g_1(x)-\frac{D}{2\ell^2}g_2(x)+M^Fg_F(x)  \\+ \frac{1}{2\pi}\left(1-M^F+\frac{DC}{2\ell^2}-M^F C_F \right)\geq 0. 
    \end{multline}

    Recalling that $M^F =1-M^L$, yields
    \begin{multline}
        M^L \left(\frac{1}{2\pi}+h_F(x)\right) \geq -\frac{D}{2}g_1(x) + \frac{D}{2\ell^2}g_2(x) \\- \frac{DC}{4\pi\ell^2} + h_F(x),
    \end{multline}
    which can be compactly rewritten as
    \begin{align}\label{eq:compact_ineq}
        M^L H(x) \geq G(x),
    \end{align}
    using \eqref{eq:HF} and \eqref{eq:G}. Inequality \eqref{eq:compact_ineq} is satisfied under a choice of $M^L$ fulfilling the constraints \eqref{eq:constraints_feasibility}, thus proving sufficiency.

   Next we prove necessity ($\Leftarrow$). Assume the problem is feasible according to Definition~\ref{def:feasibility}, i.e., there exists $\bar{\rho}^L \geq 0$ with $\int_\mathcal{S} \bar{\rho}^L \, \mathrm{d}x = M^L\in(0, 1)$ such that $\bar{\rho}^F$ is a steady-state of \eqref{eq:followers}. Then the derivation from \eqref{eq:rho_bar_l2} to \eqref{eq:compact_ineq} holds, yielding $M^L H(x) \geq G(x)$ for all $x \in \mathcal{S}$. Since $M^L$ is a constant independent of $x$, the constraints \eqref{eq:constraints_feasibility} are necessary for this inequality to hold pointwise.
\end{proof}

\begin{rem}
    Notice that $C_F$ can be expressed as
    \begin{align}
        C_F = \frac{1}{2\ell^2}\int_\mathcal{S}\left(\int \hat{v}^{FF}(x)\,\mathrm{d}x \right)\mathrm{d}x,
    \end{align}
    where we exploited the Fubini's theorem for convolutions and periodicity of circular convolutions in \eqref{eq:Cf}.
\end{rem}
\begin{rem}
    It is easy to show that $\int_\Omega G\,\mathrm{d}x = 0$. For continuity arguments, then $G$ has to switch sign over $\mathcal{S}$. Thus, feasibility is lost when $H<0$ $\forall x\in\mathcal{S}$, as it implies $\widehat{M}^L_2 < 0$ (recall that it is physically a mass).  
\end{rem}

\begin{rem}
    In the absence of inter-follower interactions $h_F = 0$ so that $H_F>0$ $\forall\,x\in\mathcal{S}$ by construction. In such a scenario, the set of constraints \eqref{eq:constraints_feasibility} reduces to the left constraint only, resembling the feasibility result presented in  \cite{maffettone2025leader}. 
\end{rem}

\begin{rem}[Physical Interpretation]\label{rem:interpretation}
The functions and constants in \eqref{eq:h} are all related to properties of the desired target distribution and intra-followers interactions:
\begin{itemize}
    \item $D g_1(x)$ represents the slope of the steady-state desired velocity field in the absence of inter-followers interactions -- see \eqref{eq:vfl_bar};
    \item $D g_2(x)$ represents the potential field associated to the desired steady-state velocity in the absence of inter-follower interactions -- see \eqref{eq:vfl_bar};
    \item the constant $C$ is proportional to the Kullback-Leibler divergence \cite{kullback1951information} between the desired and uniform distribution, thus measuring the statistical distance between the control objective and a constant density profile, that, by construction is an equilibrium of the open-loop followers' dynamics (no interactions with leaders);
    \item the function $h_F$ quantifies the influence of follower-follower interactions on the feasibility constraint.
\end{itemize}
The critical mass thresholds $\widehat{M}^L_{1, 2}$ thus emerge from the interplay between the target distribution, the diffusivity of the followers, and inter-follower interactions.
\end{rem}

\begin{rem}[Connection to herdability and shepherding] \label{rem:herdability}
Our result establishes how to tune the exact leader mass so as to achieve a given target distribution. In finite populations, this translates to a minimum number of leaders $\widehat{N}^L_{1}$ needed to control $N^F$ followers. Such a constraint is reminiscent of \emph{herdability} conditions in shepherding problems \cite{lama2024shepherding}. There, herdability quantifies the minimum number of herders (leaders) required to confine or drive a given number of targets (followers). As discussed in detail in Section~\ref{sec:deployment}, our analysis extends this concept to density shaping tasks: not only must leaders be sufficiently numerous, but the interplay between target geometry (encoded in $g_1$), diffusivity $D$, and follower-follower interactions determines whether precise spatial distributions are achievable. We provide quantitative guidance for system design in applications such as crowd management and swarm robotics.
\end{rem}

\subsection{Example} \label{sec:feasibility_example}
In this section, we provide an example that illustrates how to use Theorem~\ref{th:feasibility_theorem} in practical scenarios where we fix the desired normalized followers density $\hat{\rho}^F$. 

\begin{figure}
    \centering
    \begin{subfigure}{0.48\columnwidth}
        \centering
        \includegraphics[width=\linewidth]{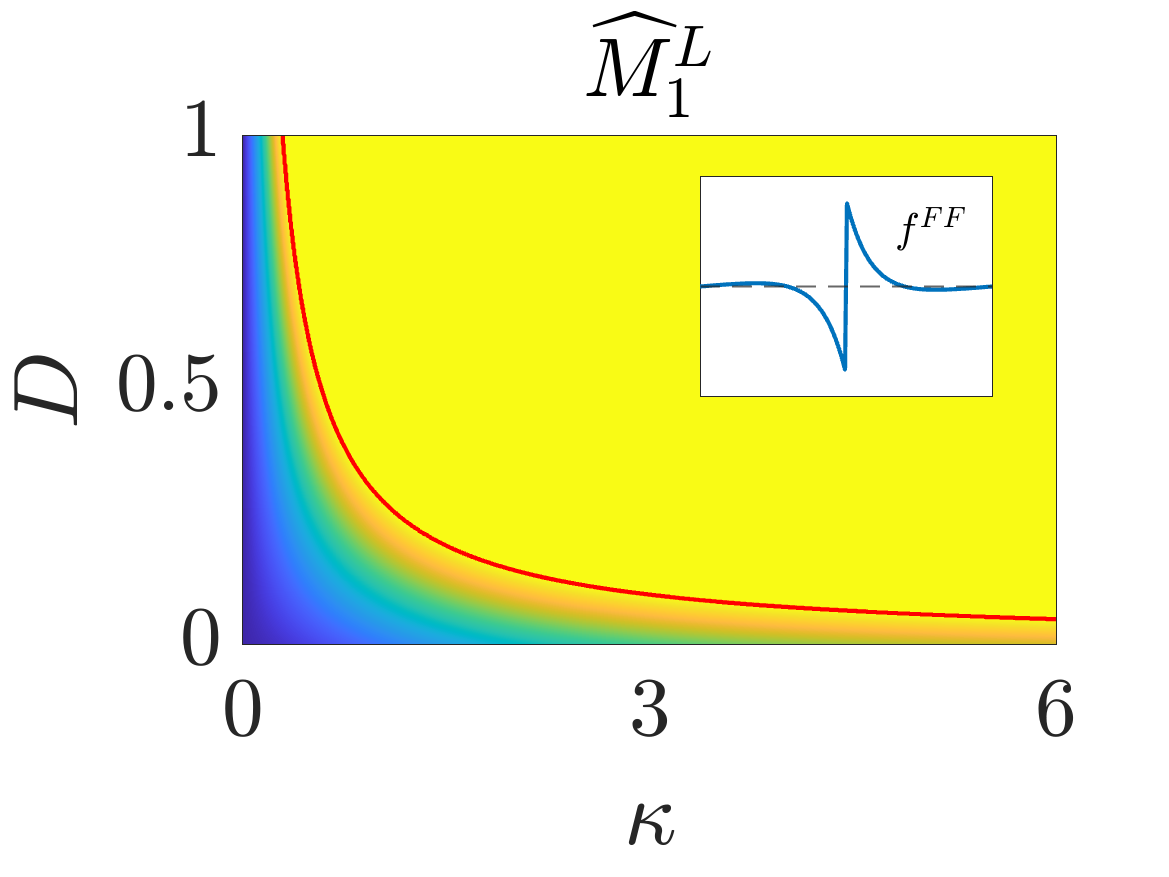}
        \caption{}
        \label{fig:M_L1_weak}
    \end{subfigure}
    \hfill
    \begin{subfigure}{0.48\columnwidth}
        \centering
        \includegraphics[width=\linewidth]{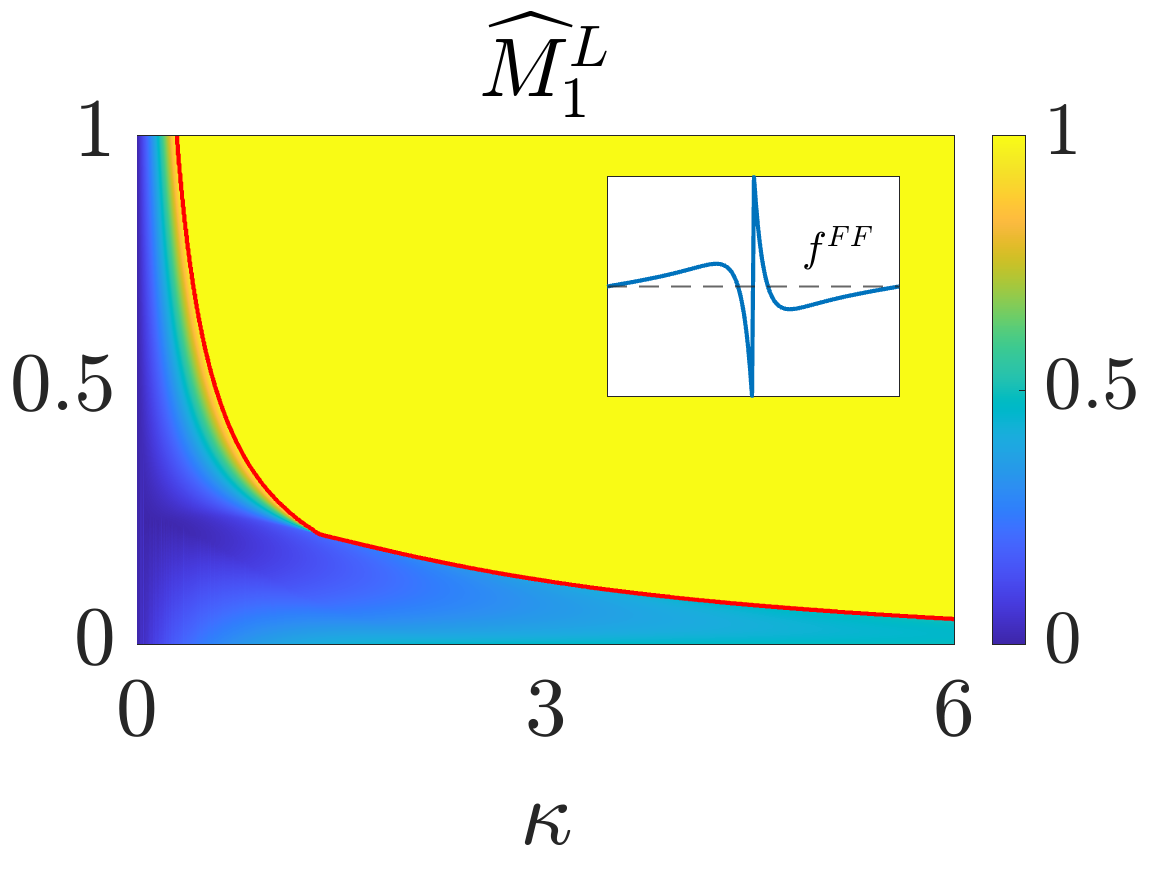}
        \caption{}
        \label{fig:M_L1_strong}
    \end{subfigure}

    \caption{Feasibility plots: minimum amount of leaders' mass $\widehat{M}^L_1$ when varying $\kappa$ and $D$ for (a) weak follower-follower Morse-type interactions and (b) strong Morse-type follower-follower interactions. The red curves in (a) and (b) indicate when feasibility is lost, while the inlet plots represent the shapes of the kernels in the two scenarios.}
    \label{fig:M_L_hat}
\end{figure}

\begin{figure}
    \centering
    \begin{subfigure}{0.48\columnwidth}
        \centering
        \includegraphics[width=\linewidth]{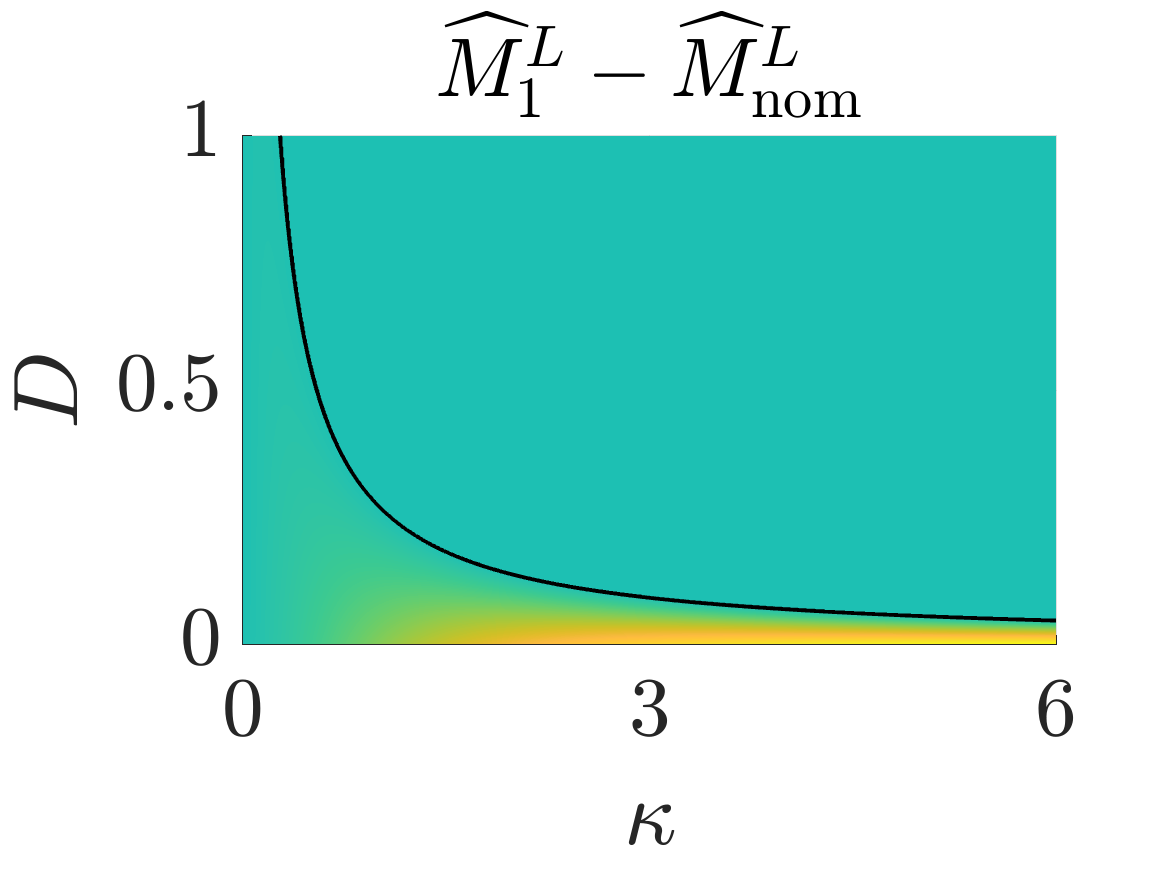}
        \caption{}    \label{fig:mass_difference_weak}
    \end{subfigure}
    \hfill
    \begin{subfigure}{0.48\columnwidth}
        \centering
        \includegraphics[width=\linewidth]{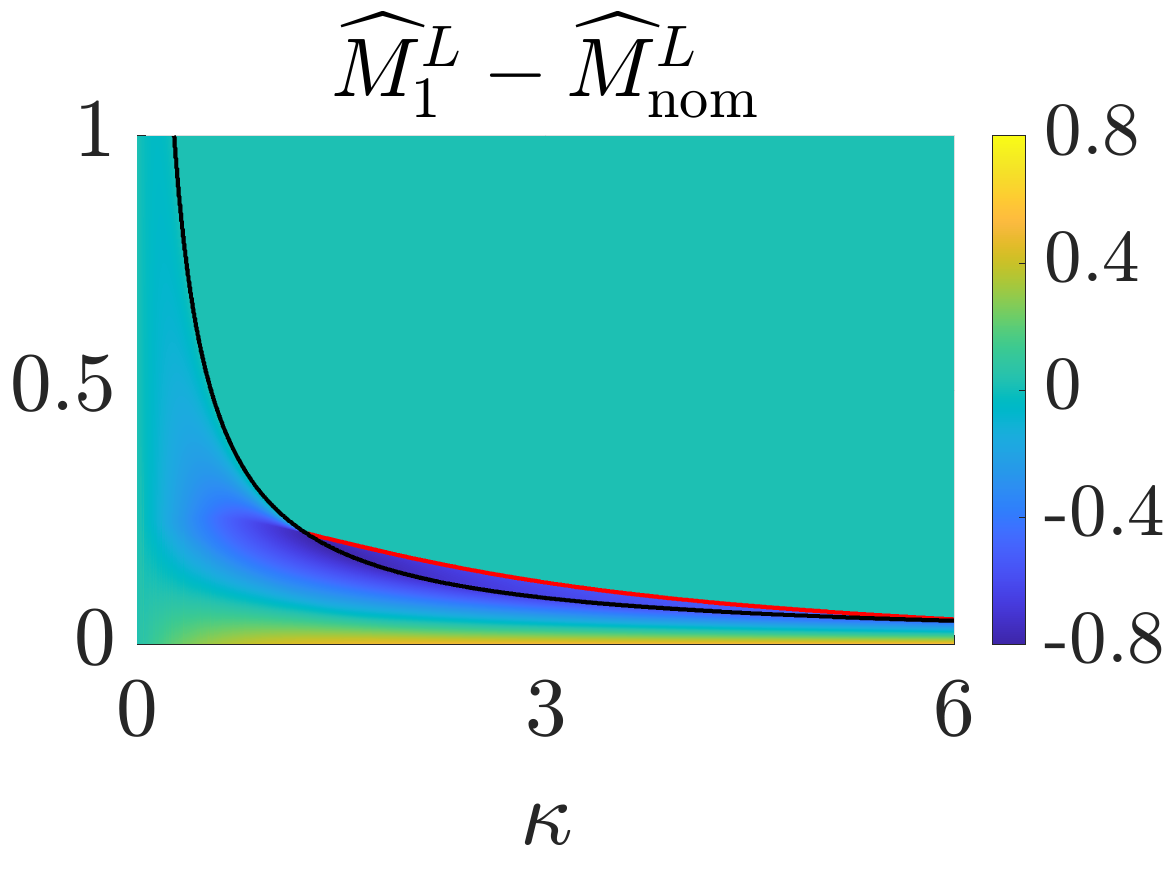}
        \caption{}
        \label{fig:mass_difference_strong}
    \end{subfigure}

    \caption{Mismatches between $\widehat{M}_1^L$ and the nominal minimum 
    leader mass $\widehat{M}^L_\mathrm{nom} := \widehat{M}_1^L\big|_{f^{FF}=0}$, 
    i.e., the minimum leader mass required to ensure feasibility in the 
    absence of follower-follower interactions. Positive values indicate that 
    inter-follower interactions require a larger leader population to achieve 
    the same target density; negative values indicate that interactions assist 
    the control task. The red and black lines denote respectively the 
    feasibility region borders in the presence or absence of follower-follower 
    interactions.}
    \label{fig:mass_difference}
\end{figure}

In particular, we consider as the normalized desired follower density the von Mises distribution
\begin{align}\label{eq:vonMises}
    \hat{\rho}^F(x) = \frac{\mathrm{e}^{\kappa \cos(x-\mu)}}{2\pi I_0(\kappa)} ,
\end{align}
where $\kappa > 0$ is the concentration coefficient, $\mu \in \mathcal{S}$ is the mean, and $I_0(\cdot)$ denotes the modified Bessel function of the first kind of order zero. For our analysis, we set $\mu = 0$, while leader-follower interactions are governed by the kernel \eqref{eq:rep_kern} with interaction length $\ell=\pi$. The inter-follower interactions are governed by the Morse kernel \eqref{eq:ff_kernel}.

We consider two simulation scenarios. In the first, we set $\ell_r=\pi/2$, $\ell_a=\pi$, and $\zeta=1$, yielding weak Morse-type interactions. In the second, we choose $\ell_r=\pi/15$, $\ell_a=\pi/2$, and $\zeta=2$, thereby increasing the attraction strength and shortening the repulsion range. Fig.~\ref{fig:M_L_hat} reports the minimum leader mass $\widehat{M}^L_1$ satisfying the feasibility constraints \eqref{eq:constraints_feasibility} in Theorem~\ref{th:feasibility_theorem}, as a function of the concentration parameter $\kappa$ of the desired followers’ distribution and the diffusion coefficient $D$. Regions where no $M^L \in (0,1)$ satisfies Theorem~\ref{th:feasibility_theorem} are highlighted in yellow. For this example, closed-form expressions of $g_1$, $g_2$, and $C$ are available (see \cite{maffettone2025leader}). In panel \ref{fig:M_L1_weak}, it holds that $\widehat{M}^L_2 \geq 1$ for all $\kappa$, so the feasibility condition reduces to $M^L \geq \widehat{M}^L_1$. In contrast, in the second scenario shown in panel \ref{fig:M_L1_strong}, there exists a region where $\widehat{M}^L_2 < 1$.

Finally, Fig.~\ref{fig:mass_difference} compares our results with those in \cite{maffettone2025leader}, where follower–follower interactions were neglected (we denote the minimum leaders mass ensuring feasibility in this case by $\widehat{M}^L_{\mathrm{nom}}$). In panel \ref{fig:mass_difference_weak}, the mismatch $\widehat{M}^L - \widehat{M}^L_{\mathrm{nom}}$ is consistently positive, indicating that inter-follower interactions require a larger leader population to achieve the same target density. In this regime, interactions act as a disturbance with respect to the desired followers’ concentration. Additionally, the feasibility regions with and without inter-follower interactions overlap. This increased control effort persists even in the case of purely repulsive follower–follower interactions (results omitted for brevity). In contrast, panel \ref{fig:mass_difference_strong} shows that for sufficiently strong inter-follower interactions there are regions where the required leader mass is smaller than in the non-interacting case (regions in panel \ref{fig:mass_difference_strong} where the mismatch $\widehat{M}^L - \widehat{M}^L_{\mathrm{nom}}$ is negative). This suggests that, when the open-loop steady-state displacement of the followers is statistically close to the target density, inter-follower interactions can assist in achieving the desired distribution. In this scenario, we also observe an increased feasibility region in the presence of inter-follower interactions: the feasibility region with interacting followers is that below the red curve in Panel \ref{fig:mass_difference_strong}, while, below the black curve we highlight the feasibility region without inter-follower interactions.

\section{Control Design} \label{sec:control_design}
\begin{figure}[t]
\centering
\resizebox{\columnwidth}{!}{%
\begin{tikzpicture}[
    >=latex,
    font=\normalsize,
    node distance=0.6cm,
    block/.style={
        draw,
        rectangle,
        minimum height=1cm,
        align=center
    },
    sum/.style={
        draw,
        circle,
        minimum size=5mm,
        inner sep=0pt
    }
]

% Blocks
\node[block, minimum width=2.6cm] (refgen)
{Leaders' reference\\from \eqref{eq:rho_bar_L_2}};

\node[sum, right=of refgen] (sum) {};

\node[block, minimum width=2.1cm, right=of sum] (controller)
{Leaders'\\controller};

\node[block, minimum width=2.6cm, right=of controller] (mixture)
{Leaders' and\\Followers' mixture};

% Input
\draw[->] ([xshift=-0.6cm]refgen.west) -- (refgen.west)
node[midway, above] {$\bar{\rho}^{F}$};

% ref -> sum
\draw[->] (refgen.east) -- (sum.west)
node[midway, above] {$\bar{\rho}^{L}$};

% Minus sign (spostato a sinistra)
\node at ([xshift=-1.5mm,yshift=-2mm]sum.south) {$-$};

% sum -> controller
\draw[->] (sum.east) -- (controller.west)
node[midway, above] {$e^{L}$};

% controller -> mixture
\draw[->] (controller.east) -- (mixture.west)
node[midway, above] {$u$};

% Outputs
\draw[->] (mixture.east) -- ++(0.6cm,0)
node[midway, above] {$\rho^{F}$};

\draw[-] ([yshift=-0.3cm]mixture.east) -- ++(0.6cm,0)
node[midway, below] {$\rho^{L}$};

% Feedback (senza freccia nel tratto verticale)
\draw
([xshift=0.6cm,yshift=-0.3cm]mixture.east)
|- ([yshift=-0.8cm]sum.south);

% Freccia solo nell'ultimo tratto verso il sommatore
\draw[->]
([yshift=-0.8cm]sum.south) -- (sum.south);

\end{tikzpicture}%
}
\caption{Block diagram of the control architecture.}
\label{fig:block_diagram}
\end{figure}
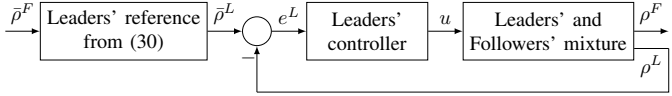

This section extends the control methodology developed in \cite{maffettone2025leader} (Section VI) to accommodate follower-follower interactions. Our approach consists of three stages: first, assuming to fulfill the feasibility constraints of Theorem~\ref{th:feasibility_theorem}, we compute $\bar{\rho}^L$ using \eqref{eq:rho_bar_L_2} -- under feasibility hypothesis, such a density integrates to $M^L$ and is positive; then, we design a control input $u$ in \eqref{eq:leaders} that drives the leader density to such a $\bar{\rho}^L$; third, we establish that achieving this leader configuration induces convergence of the follower density to the target profile $\bar{\rho}^F$. The block diagram of the proposed control architecture is given in Fig.~\ref{fig:block_diagram}.

\subsection{Leader Control Design}
To steer the leader density toward the feasible steady-state $\bar{\rho}^L$ computed in \eqref{eq:rho_bar_L_2},  we employ the feedback control law from \cite{maffettone2025leader}
\begin{align}\label{eq:control_u}
u(x, t) = -\frac{K}{\bar{\rho}^L(x) - e^L(x, t)}\int [\bar{\rho}^L(x) - \rho^L(x,t)] \,\mathrm{d}x,
\end{align}
where $K > 0$ is the control gain. This control is spatially periodic (cf. Corollary 1 in \cite{maffettone2025leader}) and preserves the total leader mass.

Under this control law, the leader density evolution admits the explicit solution
\begin{align}\label{eq:leaders_solution}
    \rho^L(x, t) = \bar{\rho}^L(x) + \Phi(x, t),
\end{align}
with
\begin{align}\label{eq:Phi}
    \Phi(x, t) = -\left[\bar{\rho}^L(x)-\rho^L_0(x)\right]\mathrm{e}^{-Kt},
\end{align}
demonstrating exponential convergence at rate $K$ from any initial condition $\rho^L_0(x)$ to the target distribution $\bar{\rho}^L(x)$.

\subsection{Follower Stability Analysis}
In this section, we study the followers' dynamics when leaders are controlled using \eqref{eq:control_u}.

We define the follower density error as
\begin{align}\label{eq:error_F}
    e^F(x, t) = \bar{\rho}^F(x) - \rho^F(x, t).
\end{align}
Substituting into the follower dynamics \eqref{eq:followers}, the error evolution satisfies 
\begin{multline}\label{eq:err_dyn}
    e^F_t(x, t) = De^F_{xx}(x, t) - D \bar{\rho}^F_{xx}(x)\\
    + \left[(\bar{\rho}^F(x) - e^F(x, t)) (v^{FL}(x, t) + v^{FF}(x, t))\right]_x,
\end{multline}
subject to the initial condition and periodic boundary conditions inherited from \eqref{eq:IC} and \eqref{eq:PBCs}, respectively.

Before establishing the stability of the follower dynamics, we present a preliminary Lemma.
\begin{lemma}\label{lemma:stability_nonlinear}
Consider the nonlinear dynamical system
\begin{align}\label{eq:eta_t}
    {\eta}_t(t) = \left(-\alpha + \beta\mathrm{e}^{-kt}\right) \eta(t) + \left(\gamma\mathrm{e}^{-kt} + \delta \eta(t)\right)\sqrt{\eta(t)},
\end{align}
where $\eta(0) \geq 0$ and $\alpha, \beta, \gamma, \delta, k > 0$ are positive constants. Then $\eta = 0$ is a locally asymptotically stable equilibrium point.
\end{lemma}
\begin{proof}
We introduce an auxiliary variable $\xi(t) = \mathrm{e}^{-kt}$ satisfying ${\xi}_t = -k\xi$ with $\xi(0) = 1$. This transforms \eqref{eq:eta_t} into the augmented system
\begin{subequations}\label{eq:aug_sys}
    \begin{align}
        {\eta}_t(t) &= (-\alpha + \beta\xi(t)) \eta(t) + (\gamma\xi(t) + \delta \eta(t))\sqrt{\eta(t)}, \label{eq:aug_eta}\\
        {\xi}_t(t) &= -k\xi(t), \label{eq:aug_xi}
    \end{align}
\end{subequations}
with initial conditions $\eta(0) \geq 0$ and $\xi(0) = 1$.

Since the $\xi$-dynamics \eqref{eq:aug_xi} is decoupled and linear, we have $\xi(t) = \mathrm{e}^{-kt} \to 0$ as $t \to \infty$. This monotonic decay precludes the existence of limit cycles in the augmented system \eqref{eq:aug_sys}.

The equilibrium points of \eqref{eq:aug_sys} are found by setting ${\eta}_t = {\xi}_t = 0$. This yields ($i$) the origin $(0, 0)$, and ($ii$) the point $(\alpha^2/\delta^2, 0)$.

To analyze stability, we compute the Jacobian matrix at each equilibrium. At the origin:
\begin{align}
    J(0,0) = \begin{bmatrix}
        -\alpha & 0 \\
        0 & -k
    \end{bmatrix},
\end{align}
which has eigenvalues $\lambda_1 = -\alpha < 0$ and $\lambda_2 = -k < 0$, confirming asymptotic stability.

At $(\alpha^2/\delta^2, 0)$, the Jacobian has one positive and one negative eigenvalue, characterizing it as a saddle point.

Therefore, the origin is locally asymptotically stable for the augmented system \eqref{eq:aug_sys}, which implies that $\eta = 0$ is locally asymptotically stable for the original system \eqref{eq:eta_t}.
\end{proof}

\begin{figure}
    \centering
    \begin{subfigure}{0.22\textwidth}
        \centering
        \includegraphics[width=\textwidth]{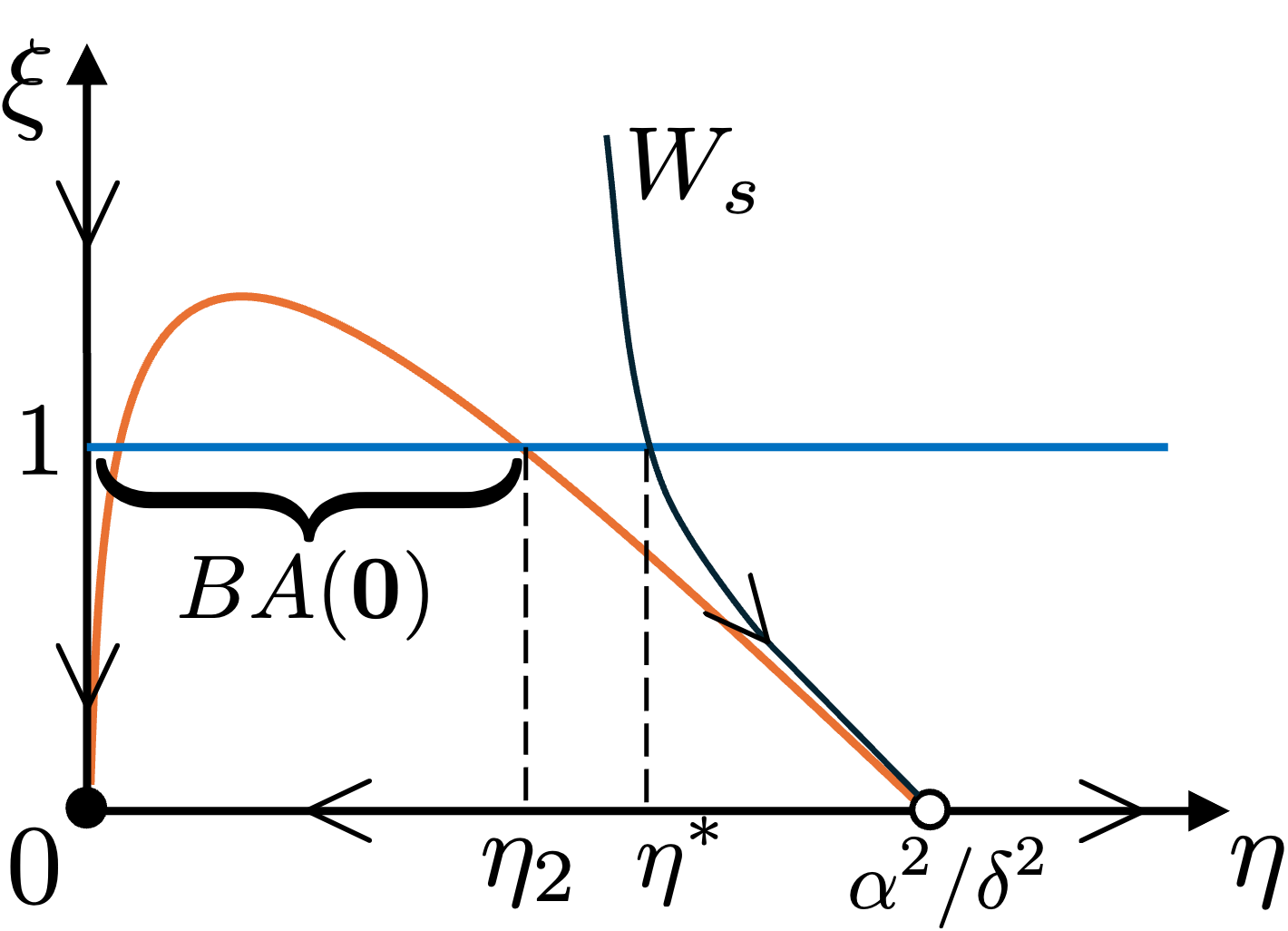}
        \caption{}
        \label{fig:ras_case1}
    \end{subfigure}
    \begin{subfigure}{0.22\textwidth}
        \centering
        \includegraphics[width=\textwidth]{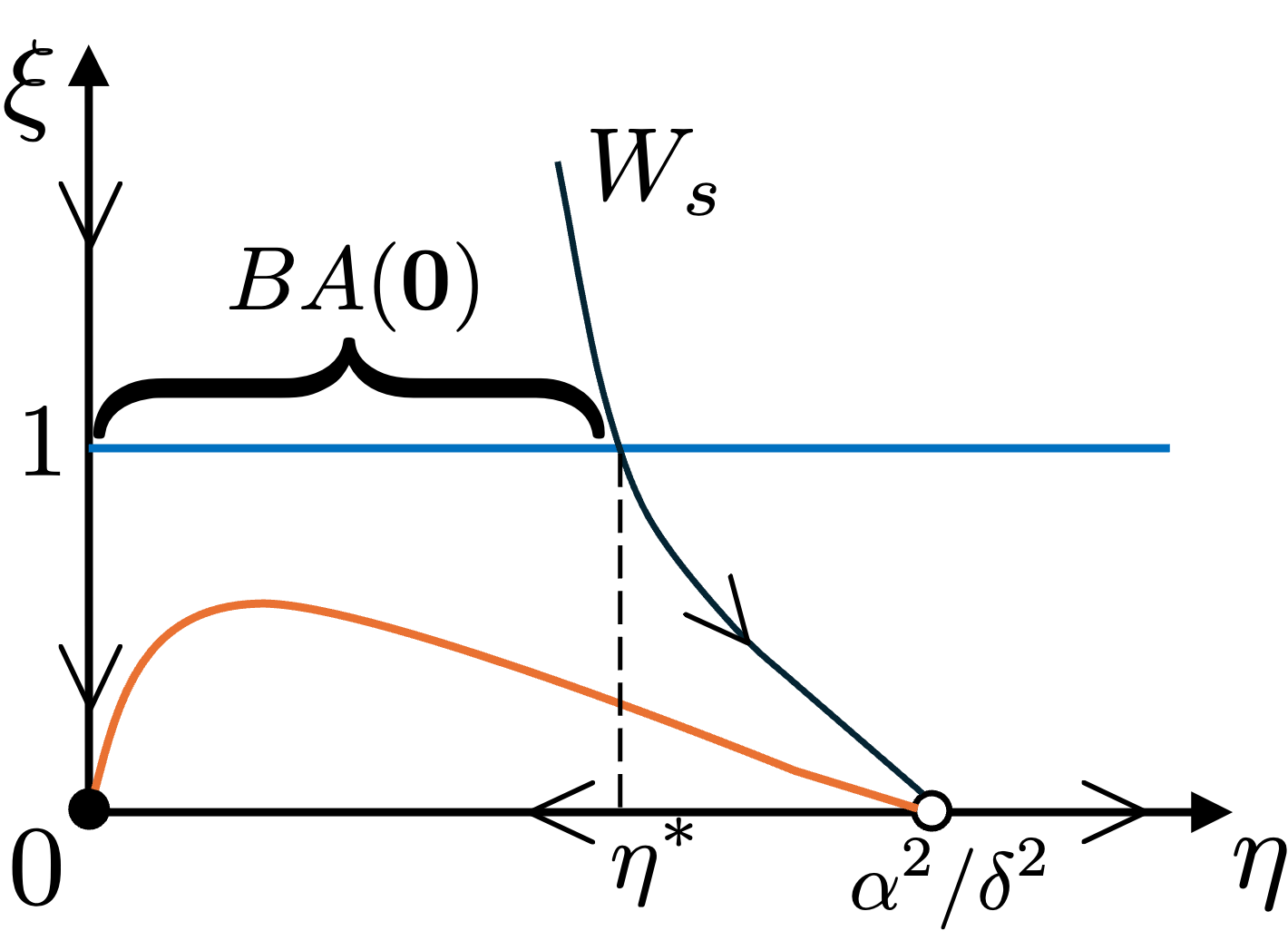}
        \caption{}
        \label{fig:ras_case2}
    \end{subfigure}
    \caption{Phase-portraits of the augmented system \eqref{eq:aug_sys} (in blue we highlight the set of possible initial conditions according to the constraint $\xi(0)=1$, in orange we show the nullcline of $\eta_t$): in (a) the nullcline intersect the axis of intial condition returning an estimation of the basin of attraction, in (b) it does not.
    }
    \label{fig:ras_estimation}
\end{figure}
\begin{rem}[Basin of Attraction]\label{rem:BA}
The stable manifold $W^s$ of the saddle point $(\alpha^2/\delta^2, 0)$ partitions the $(\eta,\xi)$-plane into two regions. Since the $\xi$-axis ($\eta = 0$) is an invariant manifold with trajectories converging to the origin, $W^s$ cannot intersect this axis.

Given the initial condition constraint $\xi(0) = 1$, the basin of attraction of the origin is characterized by
\begin{align}\label{eq:basin}
    \mathcal{B}(\mathbf{0}) = \left\{(\eta, 1) : \eta < \eta^*\right\},
\end{align}
where $\eta^*$ denotes the $\eta$-coordinate of the intersection of $W^s$ with the line $\xi = 1$.

Since an analytical expression for $W^s$ was not found, we provide a conservative estimate of $\mathcal{B}(\mathbf{0})$ using the ${\eta}_t = 0$ nullcline intersections with the initial condition line $\xi = 1$ (see Fig.~\ref{fig:ras_estimation} and Appendix~A for details).

In the fast decay limit ($k \to \infty$), the system rapidly approaches $\xi = 0$, reducing to
\begin{align}
    \dot{\eta} = -\alpha \eta + \delta \eta^{3/2}.
\end{align}
For this limiting system, all trajectories with $\eta(0) < \alpha^2/\delta^2$ converge to the origin, providing an explicit estimate of the basin of attraction.
\end{rem}

\begin{theorem}[Local Stability of Follower Dynamics]\label{thm:follower_stability}
Assume the feasibility conditions of Theorem~\ref{th:feasibility_theorem} are satisfied. If the diffusion coefficient and interaction parameters satisfy
\begin{align}\label{eq:stability_cond}
    D\left(2 - \Vert g_1\Vert_\infty \right) > F,
\end{align}
where $g_1$ is defined in \eqref{eq:g1} and 
\begin{align}\label{eq:F_def}
    F = {2(\Vert \bar{\rho}^F(\cdot)\Vert_2 \left\Vert f^{FF}_x(\cdot)\right\Vert_2 + \left\Vert \bar{\rho}^F_x(\cdot)\right\Vert_2 \Vert f^{FF}(\cdot)\Vert_2)},
\end{align}
then the error dynamics \eqref{eq:err_dyn} converges locally to zero in $\mathcal{L}^2(\mathcal{S})$, i.e., $\lim_{t\to\infty} \Vert e^F(\cdot,t)\Vert_2 = 0$ for sufficiently small initial errors.
\end{theorem}
\begin{proof}
    Substituting \eqref{eq:leaders_solution} into \eqref{eq:err_dyn}, (recalling $v^{FL} = f^{FL}*\rho^L$), we obtain
    \begin{multline}\label{eq:err_dyn2}
        e^F_t(x, t) = D(e^F_{xx}(x, t) - \bar{\rho}^F_{xx}(x))\\
        +\left[(\bar{\rho}^F(x) - e^F(x, t)) ((f^{FL}*\bar{\rho}^L))(x)+(f^{FL}*\Phi)(x, t))\right]_x\\
        +\left[(\bar{\rho}^F(x) - e^F(x, t)) v^{FF}(x, t)\right]_x.
    \end{multline}
    In feasible scenarios, we know $f^{FL}*\bar{\rho}^L=\bar{v}^{FL}$. Hence, using the expression of $\bar{v}^{FL}$ in \eqref{eq:vfl_bar} into \eqref{eq:err_dyn2}, yields
    \begin{multline}\label{eq:perturbed_err}
        e^F_t(x, t) = De^F_{xx}(x, t)+D \left(1-\mathrm{e}^{-K t}\right)\left[e^F(x, t)\frac{\bar{\rho}^F_x(x)}{\bar{\rho}^F(x)}\right]_x\\+
        \mathrm{e}^{-K t}[(\bar{\rho}^F(x)-e^F(x, t)) ((f^{FL}*\rho^L_0)(x) + (f^{FF}*\bar{\rho}^F)) \\ 
        -D\bar{\rho}^F_{x}(x)]_x
        +\left[(e^F(x, t)-\bar{\rho}^F(x)) (f^{FF}*e^F)(x, t)\right]_x,
    \end{multline}
    where we also considered \eqref{eq:Phi} and \eqref{eq:vff}. 
    
    We introduce the Lyapunov functional $\Vert e^F\Vert_2^2$, whose time derivative is $\left(\Vert e^F\Vert_2^2\right)_t = \int_\mathcal{S} e^F e^F_t\,\mathrm{d}x$. In consideration of \eqref{eq:perturbed_err}, we can write

    \begin{multline} \label{eq:Vdot_reply}
        (\Vert e^F(\cdot,t) \Vert_2^2)_t =  2D \int_\mathcal{S} e^F(x,t) e^F_{xx}(x,t) \, \mathrm{d}x + 2\mathrm{e}^{-Kt}\times \\
         \times\int_\mathcal{S} e^F(x,t)[(\bar{\rho}^F(x)-e^F(x,t))((f^{FL}*\rho^L_0)(x) +(f^{FF}*\bar{\rho}^F)(x)) \\
        -D\bar{\rho}^F_x]_x \, \mathrm{d}x - 2D(1-\mathrm{e}^{-Kt}) \int_\mathcal{S} e^F(x,t) \left[e^F(x,t) \frac{\bar{\rho}^F_x(x)} {\bar{\rho}^F(x)}\right]_x \, \mathrm{d}x \\
        +{2}\int_\mathcal{S} e^F(x,t) [(e^F(x,t)-\bar{\rho}^F(x))(f^{FF}*e^F)(x,t)]_x \, \mathrm{d}x
    \end{multline}
\begin{figure}
    \centering
    \begin{subfigure}{0.3\columnwidth}
        \centering
        \includegraphics[width=\textwidth]{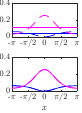}
        \caption{}
        \label{fig:initial_densitites}
    \end{subfigure}
    \hfill
    \begin{subfigure}{0.6\columnwidth}
        \centering
        \includegraphics[width=\textwidth]{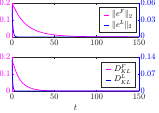}
        \caption{}
        \label{fig:errors}
    \end{subfigure}
    \caption{Regulation trial: (a) initial and final profiles of followers' (magenta) and leaders' (blue) densities, with dashed lines representing the reference distributions; (b) evolution of percentage errors and KL divergences.}
    \label{fig:validation}
\end{figure}
    
    Up to the last integral on the right hand-side, \eqref{eq:Vdot_reply} is equivalent to Equation (70) of \cite{maffettone2025leader}. Hence, we can apply the steps from (70) to (74) of \cite{maffettone2025leader} yielding the bound
    \begin{multline}\label{eq:replyVdot_bound}
        \left(\Vert e^F(\cdot, t) \Vert_2^2\right)_t \leq (-2D+D\Vert g_1(\cdot) \Vert_\infty )\Vert e^F(\cdot, t) \Vert_2^2 \\+ \Vert h_1(\cdot)\Vert_\infty \mathrm{e}^{-K t}\Vert e^F(\cdot, t) \Vert_2^2 + 2\Vert h_2(\cdot)\Vert_2 \mathrm{e}^{-K t}\Vert e^F(\cdot, t)\Vert_2 \\
        +{2}\int_\mathcal{S} e^F(x, t) \left[(e^F(x, t)-\bar{\rho}^F(x))(f^{FF}*e^F)(x, t)\right]_x\,\mathrm{d}x,
\end{multline}    
where $h_1 = (f^{FL}*\rho^L_0)_x+(f^{FF}*\bar{\rho}^F)_x$ and $h_2 =[ \bar{\rho}^F((f^{FL}*\bar{\rho}^L_0) +(f^{FF}*\bar{\rho}^F)) - D\bar{\rho}^F_x]_x$. For the last term on the right-hand side of \eqref{eq:replyVdot_bound}, it holds
\begin{subequations}\label{eq:extra_terms}
    \begin{multline}
        \int_\mathcal{S} e^F(x, t)\left[e^F(x, t) (f^{FF}*e^F)(x, t)\right]_x \,\mathrm{d}x \overset{(a)}{=} \\\overset{(a)}{=}-\int_\mathcal{S}e^F_x(x, t) e^F(x, t) (f^{FF}*e^F)(x, t)\,\mathrm{d}x \overset{(b)}{=}\\\overset{(b)}{=} -\frac{1}{2}\int_\mathcal{S}\left[(e^F(x, t))^2\right]_x (f^{FF}*e^F)(x, t)\,\mathrm{d}x\overset{(c)}{=}\\\overset{(c)}{=}\frac{1}{2}\int_\mathcal{S}(e^F(x, t))^2 (f^{FF}_x*e^F)(x, t)\,\mathrm{d}x,
    \end{multline}
    \begin{multline}
        -\int_\mathcal{S}e^F(x, t)\left[\bar{\rho}^F(x) (f^{FF}*e^F)(x, t)\right]_x \,\mathrm{d}x \overset{(d)}{=}\\\overset{(d)}{=} - \int_\mathcal{S}e^F(x, t)\bar{\rho}^F_x(x) (f^{FF}*e^F)(x, t)\,\mathrm{d}x \\- \int_\mathcal{S}e^F(x, t)\bar{\rho}^F(x) (f^{FF}_x*e^F)(x, t)\,\mathrm{d}x
    \end{multline}
\end{subequations}
where in step ($a$) we used integration by parts of periodic functions, in step ($b$) we used the identity $[(e^F)^2]_x = 2e^F_xe^F$, in step ($c$) we applied integration by parts and the definition of derivative of a convolution, and in step ($d$) we expanded the derivative of the product considering the derivative of the convolution. We can then establish the bounds on \eqref{eq:extra_terms}
\begin{subequations}\label{eq:bounds}
    \begin{multline}
        \frac{1}{2}\left\vert\int_\mathcal{S} (e^F(x, t))^2 (f^{FF}_x*e^F)(x, t)\,\mathrm{d}x\right\vert \leq\\\leq    \frac{1}{2}\int_\mathcal{S} \left\vert(e^F(x, t))^2 (f^{FF}_x*e^F)(x, t)\right\vert\,\mathrm{d}x  \overset{(e)}{=}\\\overset{(e)}{=}\frac{1}{2}\Vert (e^F(\cdot, t))^2 (f^{FF}_x *e^F)(\cdot, t)\Vert_1 \overset{(f)}{\leq} \frac{1}{2} \Vert f^{FF}_x(\cdot)\Vert_2 \Vert e^F(\cdot) \Vert_2^3,
    \end{multline}
    \vspace{-15pt}
    \begin{multline}
        \left\vert \int_\mathcal{S}e^F(x, t)\bar{\rho}^F_x(x) (f^{FF}*e^F)(x, t)\,\mathrm{d}x \right\vert \leq \\ \leq \int_\mathcal{S} \left\vert e^F(x, t)\bar{\rho}^F_x(x) (f^{FF}*e^F)(x, t)\right\vert\,\mathrm{d}x \overset{(e)}{=}\\ \overset{(e)}{=} \left\Vert e^F(\cdot, t)\bar{\rho}^F_x(\cdot) (f^{FF}*e^F)(\cdot, t)\right\Vert_1 \overset{(f)}{\leq}\\\overset{(f)}{\leq} \Vert \bar{\rho}^F_x(\cdot)\Vert_2 \Vert f^{FF}(\cdot)\Vert_2\Vert e^F(\cdot) \Vert_2^2,
    \end{multline}
    \vspace{-15pt}
    \begin{multline}
        \left\vert \int_\mathcal{S}e^F(x, t)\bar{\rho}^F(x) (f^{FF}_x*e^F)(x, t)\,\mathrm{d}x \right\vert \leq \\ \leq \int_\mathcal{S} \left\vert e^F(x, t)\bar{\rho}^F(x) (f^{FF}_x*e^F)(x, t)\right\vert\,\mathrm{d}x \overset{(e)}{=}\\ \overset{(e)}{=} \left\Vert e^F(\cdot, t)\bar{\rho}^F(\cdot) (f^{FF}_x*e^F)(\cdot, t)\right\Vert_1 \overset{(f)}{\leq}\\\overset{(f)}{\leq} \Vert \bar{\rho}^F(\cdot)\Vert_2 \Vert f^{FF}_x(\cdot)\Vert_2\Vert e^F(\cdot) \Vert_2^2
    \end{multline}
\end{subequations}
where in step ($e$) we used the definition of $\mathcal{L}^1$ norm, in step ($f$) we used the H\"older's inequality and the Young's inequality for convolutions. Using the bounds \eqref{eq:bounds} into \eqref{eq:replyVdot_bound} yields
\begin{multline}\label{eq:final_bound}
    \left(\Vert e^F(\cdot, t) \Vert_2^2\right)_t \leq (-2D+D\Vert g_1(\cdot) \Vert_\infty + F)\Vert e^F(\cdot, t) \Vert_2^2 \\+ \Vert h_1(\cdot)\Vert_\infty \mathrm{e}^{-K t}\Vert e^F(\cdot, t) \Vert_2^2 + 2\Vert h_2(\cdot)\Vert_2 \mathrm{e}^{-K t}\Vert e^F(\cdot, t)\Vert_2\\+ \Vert f^{FF}_x(\cdot)\Vert_2 \Vert e^F(\cdot,t) \Vert_2^3,
\end{multline}
where $F$ is defined in \eqref{eq:F_def}. The bounding system in \eqref{eq:final_bound} is in the form discussed in Lemma \ref{lemma:stability_nonlinear} (with $\eta = \Vert e^F\Vert_2^2$, $\alpha = \vert -2D + D\Vert g_1 \Vert_\infty + F\vert$, $\beta = \Vert h_1\Vert_\infty$, $\gamma = 2\Vert h_2\Vert_2$ and $\delta = 1/2\Vert f^{FF}_x\Vert_2$, $k = K$), so that we can conclude about the local stability of the error dynamics in $\mathcal{L}^2(\mathcal{S})$ by means of the comparison principle \cite{khalil2002nonlinear}. Note that the reasoning in Remark \ref{rem:BA} holds for the basin of attraction of the origin. %Note that, exploiting Lemma \ref{lemma:lemma1}, it is also possible to explicitly recover the basin of attraction of the origin.
\end{proof}

\begin{rem}
When follower-follower interactions are absent ($F = 0$), condition \eqref{eq:stability_cond} reduces to $\|g_1\|_\infty < 2$, recovering the stability condition derived in \cite{maffettone2025leader}. The presence of cohesive interactions ($F > 0$) tightens this bound to $\|g_1\|_\infty < 2 - F/D$, quantifying how inter-follower attraction reduces the margin for stable density control. Intuitively, cohesive followers resist spatial reorganization by the leaders, requiring either stronger diffusion or less concentrated target distributions.
\end{rem}

\begin{rem}
The feasibility conditions \eqref{eq:constraints_feasibility} can be 
verified offline given the target density $\bar{\rho}^F$. The control 
law \eqref{eq:control_u} requires only pointwise evaluation of $\rho^L$, enabling distributed implementation via 
local density estimation \cite{dilorenzo2025decentralized}.
\end{rem}

\section{Numerical Validation}\label{sec:numerical_validation}
\begin{figure*}
    \centering
    \begin{subfigure}{0.32\textwidth}
        \centering
        \includegraphics[width=\textwidth]{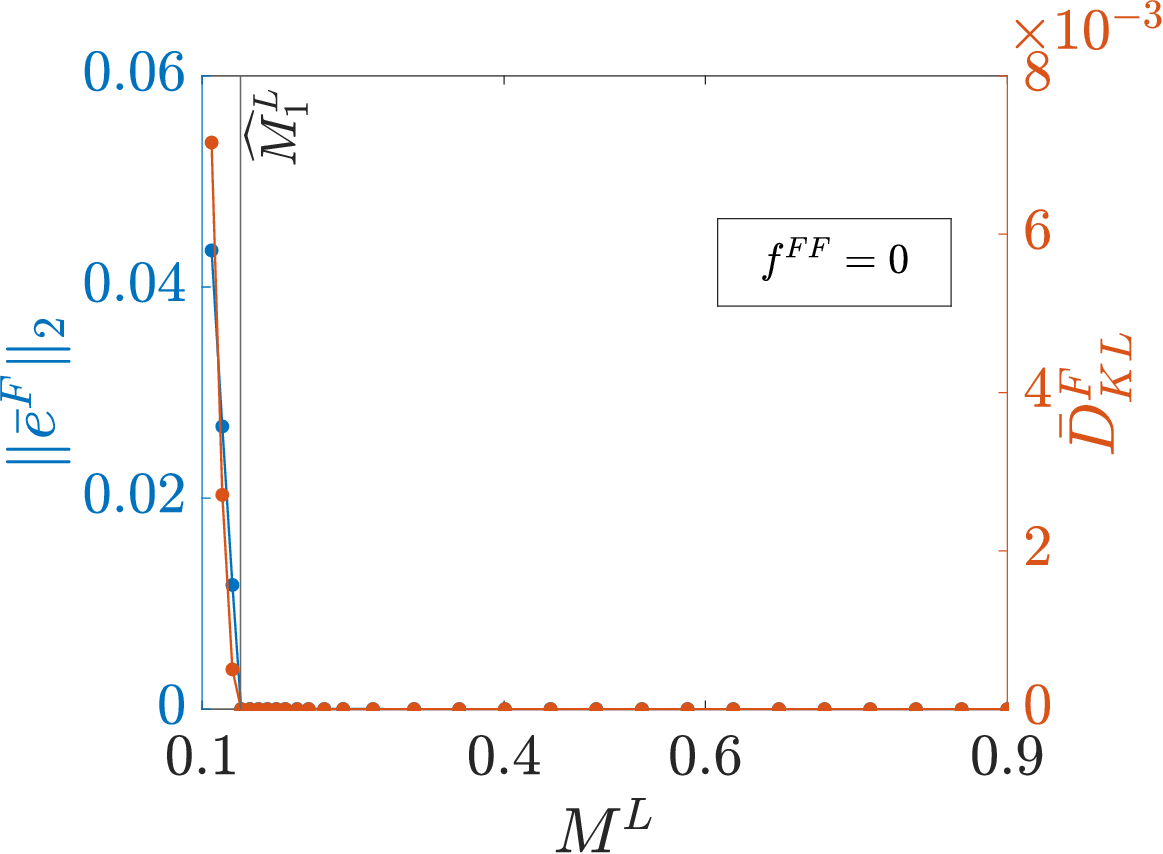}
        \caption{}
        \label{ffig:macro_ml_strong_no}
    \end{subfigure}
    \hfill
    \begin{subfigure}{0.32\textwidth}
        \centering
        \includegraphics[width=\textwidth]{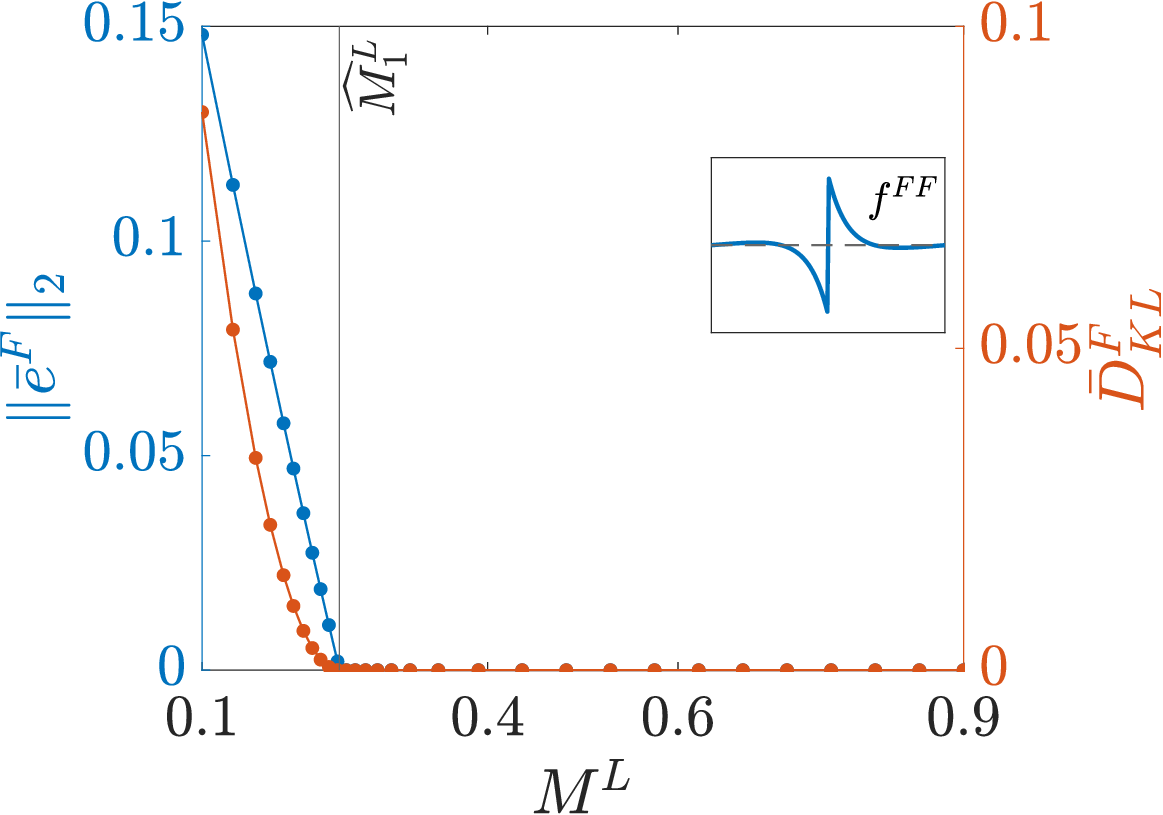}
        \caption{}
        \label{fig:macro_ml_weak_ff}
    \end{subfigure}
    \hfill
    \begin{subfigure}{0.32\textwidth}
        \centering
        \includegraphics[width=\textwidth]{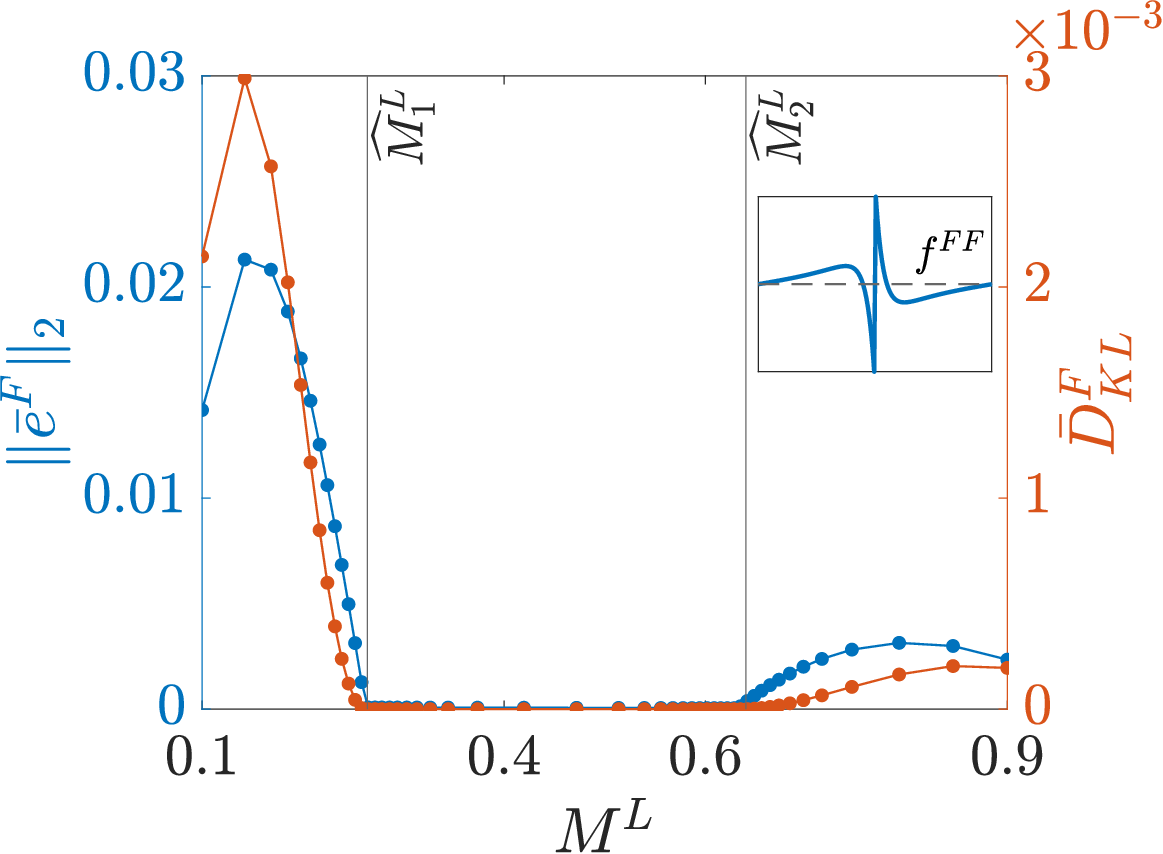}
        \caption{}
        \label{fig:macro_ml_strong_ff}
    \end{subfigure}
    \caption{Steady-state followers' error as a function of the leaders' mass in (a) absence, (b) weak and (c) strong follower-follower interactions.}
    \label{fig:varying_ml}
\end{figure*}
We validate the proposed control strategy through numerical simulations of system \eqref{eq:themodel}. For the numerical implementation, we discretize the PDEs using central differences on a spatial mesh of 500 points and employ forward Euler time integration with step size $\Delta t = 0.01$ over a time horizon of $T = 150$ units.

The target follower density is chosen as a von Mises distribution \eqref{eq:vonMises} with mean $\mu = 0$ and concentration parameter $\kappa = 1$. The leader-follower interactions are modeled using the repulsive kernel with length scale $\ell = \pi$, while follower-follower interactions employ the Morse-type kernel \eqref{eq:ff_kernel} with parameters $\ell_r = \pi/2$, $\ell_a = \pi$, and $\zeta = 1$. We set the diffusion coefficient to $D = 0.02$ and the control gain to $K = 1$. The mass distribution between leaders and followers is chosen as $M^F = 0.75$ and $M^L = 0.25$, respectively, which satisfies the feasibility conditions established in Theorem~\ref{th:feasibility_theorem}.

To quantify performance of the proposed strategy, we monitor $\|e^i\|_2$, for $i=L, F$, and, as an extra performance index, we also consider the Kullback-Leibler (KL) divergence \cite{kullback1951information} between the desired leaders' (followers') density and the leaders' (followers') density, that is
\begin{align} \label{eq:KL_divergence}
    D_{KL}^i(t) = \int_\Omega \rho^i(x,t) \log\left(\frac{\rho^i(x,t)}{\bar{\rho}^i(x)}\right) \, \mathrm{d} x, \quad i\in\{L,F\}.
\end{align}

Fig.~\ref{fig:validation} presents the simulation results, demonstrating that both leader and follower errors $\| e^F \|_2$ and $\| e^L \|_2$ decay asymptotically to zero, thereby confirming the theoretical convergence guarantees established in our stability analysis. We remark that convergence is attained also choosing initial conditions consistently far away from the desired equilibrium, pointing at the conservativeness of our estimate of the basin of attraction.

For validating the estimates derived in Theorem~\ref{th:feasibility_theorem} for the leaders' mass ensuring feasibility $\widehat{M}^L_{1,2}$, we conduct a simulation campaign by varying $M^L$ in the interval $[0.1, 0.9]$ (and, coherently, $M^F$), while fixing $\bar{\rho}^F$ as the von Mises distribution in \eqref{eq:vonMises}, with $\mu$ and $\kappa$ denoting its mean and concentration coefficient. We initialize the followers' and leaders' densities at their steady states, $\bar{\rho}^F$ and $\bar{\rho}^L$, respectively, and let the simulation run for $T=100$ time units\footnote{In regions where feasibility is not guaranteed by Theorem~\ref{th:feasibility_theorem}, a positive leaders' density integrating to $M^L$ that makes $\bar{\rho}^F$ a steady state for the followers does not exist. In such scenarios, we choose $B$ in \eqref{eq:rho_bar_l} so that $\min_{x}\bar{\rho}^L=0$ and re-normalize $\bar{\rho}^L$ so that it integrates to $M^L$.}. Controlling the leaders via \eqref{eq:control_u}, we record $\Vert \bar{e}^F\Vert_2=\Vert {e}^F(\cdot, T)\Vert_2$ and $\bar{D}_{KL}^F = D_{KL}^F(T)$ while varying $M^L$ under three different settings: (a) no follower-follower interactions (other parameters: $\mu = 0$, $\kappa = 1$, $D=0.04$, $K = 1$); (b) weak follower-follower Morse interactions, i.e., $\ell_r = \pi/2$, $\ell_a = \pi$ and $\alpha = 1$ (other parameters: $\mu = 0$, $\kappa = 1$, $D=0.02$, $K = 1$); (c) strong follower-follower Morse interactions, i.e., $\ell_r = \pi/15$, $\ell_a = \pi/2$ and $\alpha = 2$ (other parameters: $\mu = 0$, $\kappa = 2$, $D=0.16$, $K = 1$). The results for these three settings are reported in Fig.~\ref{fig:varying_ml}, confirming the accuracy of our theoretical estimates; note that we sample more values of $M^L$ near the feasibility thresholds. In particular, in the first scenario (no inter-follower interactions, cf. Fig.~\ref{ffig:macro_ml_strong_no}), we observe that above the threshold $\widehat{M}^L_1 \approx 0.14$ the error is consistently zero; moreover, $\widehat{M}^L_2>1$. A similar behavior is observed with weak inter-follower interactions (cf. Fig.~\ref{fig:macro_ml_weak_ff}), with $\widehat{M}^L_1 \approx 0.24$. Finally, with strong inter-follower interactions (cf. Fig.~\ref{fig:macro_ml_strong_ff}), we obtain $\widehat{M}^L_1\approx 0.25$ and $\widehat{M}^L_2\approx 0.63$: when $\widehat{M}^L_1\leq M^L\leq\widehat{M}^L_2$ the error remains zero; otherwise, as expected from the analysis, followers cannot maintain the profile described by $\bar{\rho}^F$, resulting in non-zero steady-state errors.

{
\section{Extension to higher-dimensions}\label{sec:ext_hd}
Our results can be straightforwardly extended to higher dimensional periodic domains $\Omega = [-\pi, \pi]^d$ for $d=2,3$. The model \eqref{eq:themodel} becomes
\begin{subequations}\label{eq:hd_model}
    \begin{align}
        \rho_t^L(\mathbf{x},t) + \nabla\cdot\left[\rho^L(\mathbf{x}, t) \mathbf{u}(\mathbf{x},t)\right] = 0,\label{eq:leaders_hd}
    \end{align}
    \begin{multline}
        \rho_t^F(\mathbf{x},t) + \nabla\cdot\left[\rho^F\left(\mathbf{v}^{FL}(\mathbf{x},t) + \mathbf{v}^{FF}(\mathbf{x},t)\right)\right] = \\ D \nabla^2\rho^F(\mathbf{x},t),\label{eq:followers_hd}
    \end{multline}
\end{subequations}
where $\mathbf{x}\in\Omega$, $\nabla\cdot$ and $\nabla^2$ represent the divergence and Laplacian operators, $\mathbf{u}$ is a control input to design. Furthermore
\begin{subequations}
    \begin{align}
        \mathbf{v}^{FL} &= (\mathbf{f}^{FL}*\rho^L)(\mathbf{x}, t),\\
        \mathbf{v}^{FF} &= (\mathbf{f}^{FF}*\rho^F)(\mathbf{x}, t),
    \end{align}
\end{subequations}
are velocity fields characterizing the followers advection ($\mathbf{f}^{FL}$ and $\mathbf{f}^{FF}$ are the $d$-dimensional version of the periodic interaction kernels $f^{FL}$ and $f^{FF}$\footnote{For $d = 2, 3$ a closed form for such periodic kernels is not available; hence, they need to be expressed through an infinite series then truncated for implementation purposes -- see \cite{boldini2024stigmergy} for more details.}). The convolution is to be interpreted as the $d$-dimensional circular convolution. As in the one-dimensional case, the system is complemented with initial conditions and periodic boundary conditions for the densities on $\partial\Omega$.

\subsection{Feasibility analysis} \label{sec:feasibility_hd}
In higher-dimensions, the feasibility analysis (see Definition \ref{def:feasibility}) starts by considering \eqref{eq:followers_hd} at steady-state, that is
\begin{align}\label{eq:ss_hd}
    \nabla\cdot\left[\bar{\rho}^F(\mathbf{x})\left(\bar{\mathbf{v}}^{FL}(\mathbf{x})+\bar{\mathbf{v}}^{FF}(\mathbf{x})\right)\right] = D\nabla^2\bar{\rho}^F(\mathbf{x}),
\end{align}
where $\bar{\mathbf{v}}^{FL} = (\mathbf{f}^{FL}*\bar{\rho}^L)$ and $\bar{\mathbf{v}}^{FF} = (\mathbf{f}^{FF}*\bar{\rho}^F)$. As \eqref{eq:ss_hd} is only a scalar relation, it does not suffice to uniquely determine $\mathbf{v}^{FL}$. Hence, following the approach in \cite{maffettone2025leader} we pose $\mathbf{w} = \bar{\rho}^F (\bar{\mathbf{v}}^{FL} + \bar{\mathbf{v}}^{FF})$ and close the problem with an irrotationality constraint, that is
\begin{align}\label{eq:flux_relation}
    \begin{cases}
        \nabla \cdot \mathbf{w}(\mathbf{x}) = D\nabla^2\bar{\rho}^F(\mathbf{x}),\\
        \nabla\times\mathbf{w}(\mathbf{x}) = 0,
    \end{cases}
\end{align}
with periodic boundary conditions. As $\Omega$ is simply connected, \eqref{eq:flux_relation} can be solved by introducing a scalar potential such that $\mathbf{w} = -\nabla\varphi$. Substituting the scalar potential into \eqref{eq:flux_relation} yields
\begin{align}
    \nabla^2\varphi(\mathbf{x}) = -D\nabla^2\bar{\rho}^F(\mathbf{x}),
\end{align}
which is fulfilled if 
\begin{align} \label{eq:v_FL_hd}
    \bar{\mathbf{v}}^{FL}(\mathbf{x}) = D\frac{\nabla\bar{\rho}^F(\mathbf{x})}{\bar{\rho}^F(\mathbf{x})} - \bar{\mathbf{v}}^{FF}(\mathbf{x}).
\end{align}
Note that $\bar{\mathbf{v}}^{FL}$ naturally resembles its one-dimensional counterpart in \eqref{eq:vfl_bar}.

Recalling $\bar{\mathbf{v}}^{FL} = (\mathbf{f}^{FL}*\bar{\rho}^L)$, we recover $\bar{\rho}^L$ via deconvolution, that is
\begin{align}\label{eq:deconv_hd}
    \bar{\rho}^L(\mathbf{x}) = R(\mathbf{x}) + A,
\end{align}
where $A$ is an arbitrary constant. Since, in higher-dimensions, a closed form for the repulsive periodic kernel $\mathbf{f}^{FL}$ was not found, $R$ needs to be computed numerically \cite{di2025continuification}. The deconvolution is defined up to a constant because of the linearity of the convolution operator and since interaction kernels are assumed to be odd-periodic on $\Omega$. The feasibility problem is recast into that of seeking $A$ such that 
\begin{subequations}
\begin{align}
    \bar{\rho}^L(\mathbf{x}) &\geq 0,\\
    \int_\Omega\bar{\rho}^L(\mathbf{x})\,\mathrm{d}\mathbf{x} &= M^L.
\end{align}
\end{subequations}
Such constraints can be numerically evaluated to check if feasibility is guaranteed. For instance, setting $A = -\min_\mathbf{x} R(\mathbf{x})$ in \eqref{eq:deconv_hd}, returns the non negative $\bar{\rho}^L$ with the smallest possible integral; if such an integral is less or equal then $1-M^F$, then the problem is feasible.

\subsection{Control design}
Here we extend the control design procedure described in Section \ref{sec:control_design} to higher dimensional domains.
\subsubsection{Leaders' control}
To steer the leaders' density towards the desired distribution $\bar{\rho}^L$, we choose
\begin{align}
    \nabla \cdot [\rho^L(\mathbf{x},t)\mathbf{u}(\mathbf{x},t)] = -K e^L(\mathbf{x},t),
\end{align}
with $K>0$, to ensure
\begin{align}
    e^L_t(\mathbf{x},t) = -K e^L(\mathbf{x},t),
\end{align}
establishing exponential point-wise convergence of the error to 0. To obtain $\mathbf{u}$ explicitly, we add the irrotationality condition
\begin{align}
    \nabla \times [\rho^L(\mathbf{x},t)\mathbf{u}(\mathbf{x},t)] = 0,
\end{align}
and solve the associated Poisson problem using Fourier series, as in \cite{maffettone2024mixed}. This control is spatially periodic and ensures conservation of mass (cfr. \cite{maffettone2025leader}, Section X.A). Under this control law, as in the one dimensional case, the leaders' density satisfies
\begin{align}
    \rho^L(\mathbf{x},t) = \bar{\rho}^L(\mathbf{x},t) + \Phi(\mathbf{x},t),
\end{align}
with
\begin{align}
    \Phi(\mathbf{x},t) = -[\bar{\rho}^L(\mathbf{x}) - \rho_0^L(\mathbf{x})] \mathrm{e}^{-Kt}.
\end{align}

\subsubsection{Followers' stability analysis}
Considering the followers' error defined in \eqref{eq:error_F}, its dynamics satisfies
\begin{multline} \label{eq:err_dyn_multidim}
    e^F_t(\mathbf{x},t) = D \nabla^2 e^F(\mathbf{x},t) - D(1-\mathrm{e}^{-Kt}) \nabla \cdot \left[e^F \frac{\nabla \bar{\rho}^F(\mathbf{x})}{\bar{\rho}^F(\mathbf{x})}\right] \\ + \mathrm{e}^{-Kt} \nabla \cdot[(\bar{\rho}^F(\mathbf{x}) - e^F(\mathbf{x},t))((\mathbf{f}^{FF}*\bar{\rho}^F)(\mathbf{x}) + (\mathbf{f}^{FL}*\rho_0^L)(\mathbf{x})) \\ - D \nabla \bar{\rho}^F(\mathbf{x})]  - \nabla \cdot [(\bar{\rho}^F(\mathbf{x}) - e^F(\mathbf{x},t) (\mathbf{f}^{FF}*e^F)(\mathbf{x},t) ]
\end{multline}
\begin{theorem}\label{th:foll_stab_hd}
    Let the control problem be feasible according to Definition \ref{def:feasibility}, and define 
    \begin{align}
        G_1(\mathbf{x}) = \nabla \cdot \left[\frac{\nabla \bar{ \rho}^F}{\bar{\rho}^F}\right].
    \end{align}
    If the diffusion coefficient and interaction parameters satisfy
    \begin{align}
        D(2- \| G_1(\cdot)\|_\infty)>F,
    \end{align}
    where 
    \begin{align} \label{eq:F_multidim}
        F = 2\left(\| \bar{\rho}^F(\cdot) \|_2 \| \nabla \cdot \mathbf{f}^{FF}(\cdot)\|_2 + \| \nabla \bar{\rho}^F(\cdot) \|_2 \|\mathbf{f}^{FF}(\cdot)\|_2\right),
    \end{align}
    then the error dynamics \eqref{eq:err_dyn_multidim} converges locally to zero in $\mathcal{L}^2(\Omega)$.
\end{theorem}
\begin{proof}
    Choosing $\|e^F\|_2^2$ as a Lyapunov functional for \eqref{eq:err_dyn_multidim}, we obtain
    \begin{multline} \label{eq:err_dyn_multidim_2}
        \left(\|e^F(\cdot,t)\|_2^2\right)_t = 2D \int_\Omega e^F(\mathbf{x},t) \nabla ^2e^F(\mathbf{x},t) \, \mathrm{d}\mathbf{x} \\ - 2D(1 - \mathrm{e}^{-Kt}) \int_\Omega e^F(\mathbf{x},t) \nabla \cdot\left[ e^F(\mathbf{x},t) \frac{\nabla \bar{\rho}^F(\mathbf{x})}{\bar{\rho}^F(\mathbf{x})} \right] \, \mathrm{d}\mathbf{x} \\ - 2 \mathrm{e}^{-Kt} \int_\Omega e^F(\mathbf{x},t) \nabla \cdot [D \nabla \bar{\rho}^F \\ -\left(\bar{\rho}^F(\mathbf{x})-e^F(\mathbf{x},t)\right) \left( (\mathbf{f}^{FF}*\bar{\rho}^F)(\mathbf{x}) + (\mathbf{f}^{FL}*\rho_0^L)(\mathbf{x})\right) ] \, \mathrm{d} \mathbf{x} \\ -2 \int_\Omega  e^F(\mathbf{x},t) \nabla \cdot \left[ (\bar{\rho}^F\left(\mathbf{x}) - e^F(\mathbf{x},t)\right)(\mathbf{f}^{FF}*e^F)(\mathbf{x},t) \right] \, \mathrm{d} \mathbf{x} 
    \end{multline}
    Up to the last integral on the right-hand side, \eqref{eq:err_dyn_multidim_2} is equivalent to Equation (117) of \cite{maffettone2025leader} (upon setting $\alpha=0$). Hence, utilizing vectorial identities and the divergence theorem, we can establish the bounds
    \begin{multline} \label{eq:err_bounds_1}
        \left(\|e^F(\cdot,t)\|_2^2\right)_t \leq (-2D + D\| G_1(\cdot) \|_\infty)\| e^F(\cdot,t) \|_2^2 \\ + \| H_1(\cdot) \|_\infty \mathrm{e}^{-Kt} \| e^F(\cdot,t) \|_2^2 + 2 \| H_2(\cdot) \|_2 e^{-Kt} \| e^F(\cdot,t) \|_2 \\ + \int_\Omega e^F(\mathbf{x},t) \nabla \cdot \left[  (e^F(\mathbf{x},t) - \bar{\rho}^F(\mathbf{x}))(\mathbf{f}^{FF}*e^F)(\mathbf{x},t)\right] \, \mathrm{d} \mathbf{x},
    \end{multline}
    where $H_1 = \nabla \cdot \left[(\mathbf{f}^{FF}*\bar{\rho}^F) + (\mathbf{f}^{FL}*\rho_0^L)\right]$ and $H_2 = \nabla \cdot \left[\bar{\rho}^F\left((\mathbf{f}^{FF}*\bar{\rho}^F) + (\mathbf{f}^{FL}*\rho_0^L)\right) - D \nabla \bar{\rho}^F\right]$. The last integral in \eqref{eq:err_bounds_1} can be rewritten as in \eqref{eq:extra_terms}, and exploiting bounds similar to those derived in \eqref{eq:bounds}, we establish
    \begin{multline}
        \left(\|e^F(\cdot,t)\|_2^2\right)_t \leq (-2D +D\| G_1(\cdot) \|_\infty + F) \| e^F(\cdot,t)\|_2^2 \\ + \|H_1(\cdot)\|_\infty \mathrm{e}^{-Kt} \| e^F(\cdot,t) \|_2^2 + 2\|H_2(\cdot)\|_2 \mathrm{e}^{-Kt} \| e^F(\cdot,t) \|_2 \\ + \| \nabla \cdot \mathbf{f}^{FF}(\cdot) \|_2 \| e^F(\cdot,t) \|_2^3,
    \end{multline}
    with $F$ defined in \eqref{eq:F_multidim}. The bounding system is in the form discussed in Lemma \ref{lemma:stability_nonlinear}, proving the theorem.
\end{proof}

\subsection{Numerical validation of the macroscopic controller} \label{sec:macroscopic_validation_hd}

\begin{figure*}
    \centering
    \begin{subfigure}{0.3\textwidth}
        \centering
        \includegraphics[width=\textwidth]{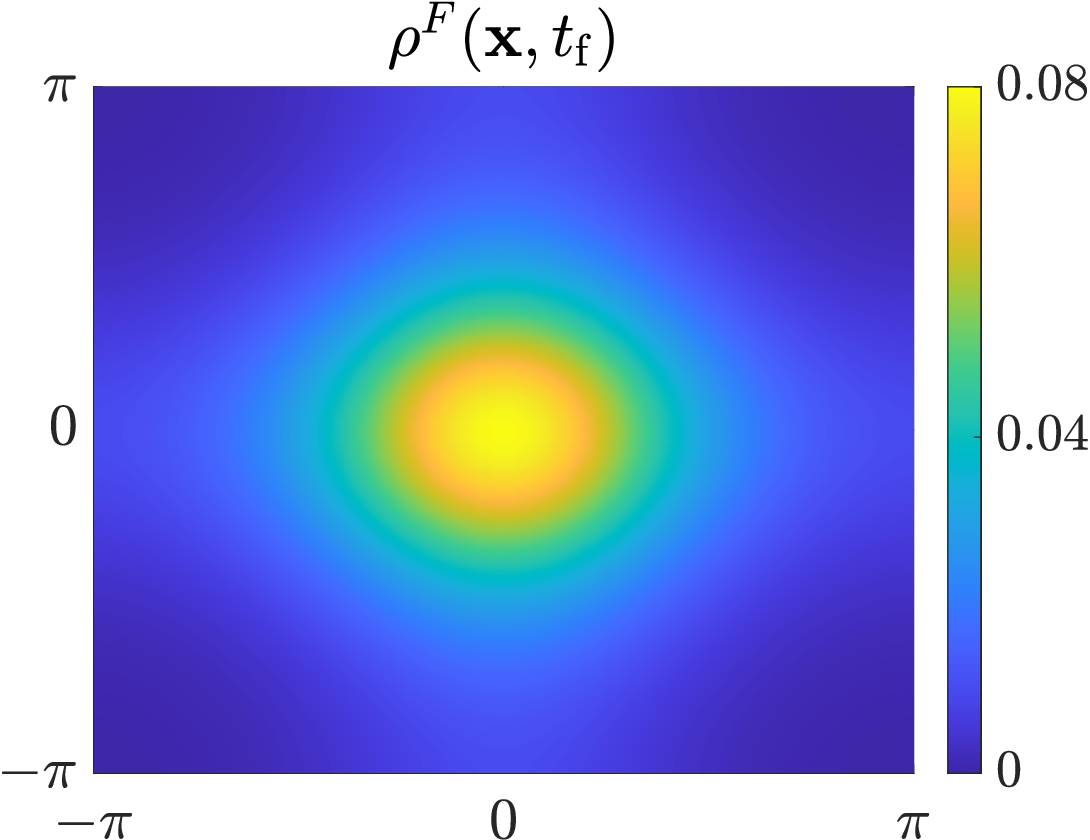}
        \caption{}
        \label{fig:initial_densitites_hd}
    \end{subfigure}
    % \hspace{0.1\textwidth}
    \hfill
    \begin{subfigure}{0.3\textwidth}
        \centering
        \includegraphics[width=\textwidth]{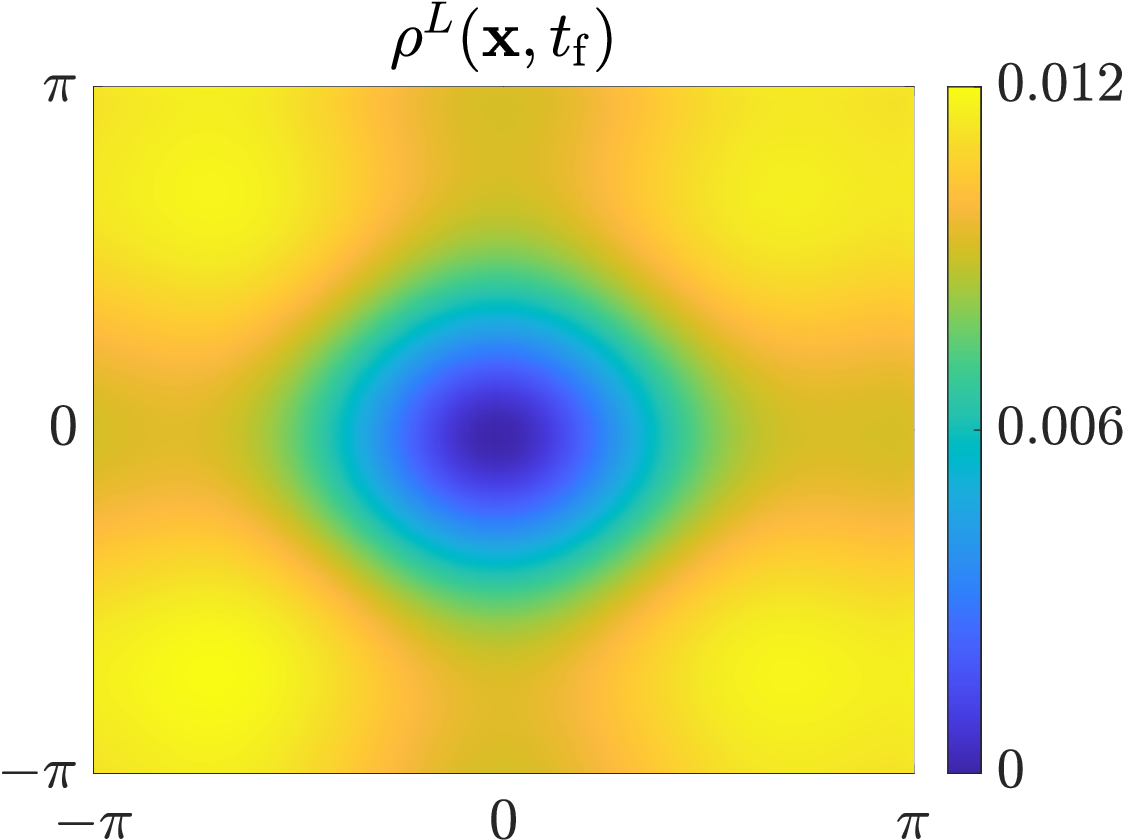}
        \caption{}
        \label{fig:final_densities_hd}
    \end{subfigure}
    \hfill
    \begin{subfigure}{0.3\textwidth}
        \centering
        \includegraphics[width=\textwidth]{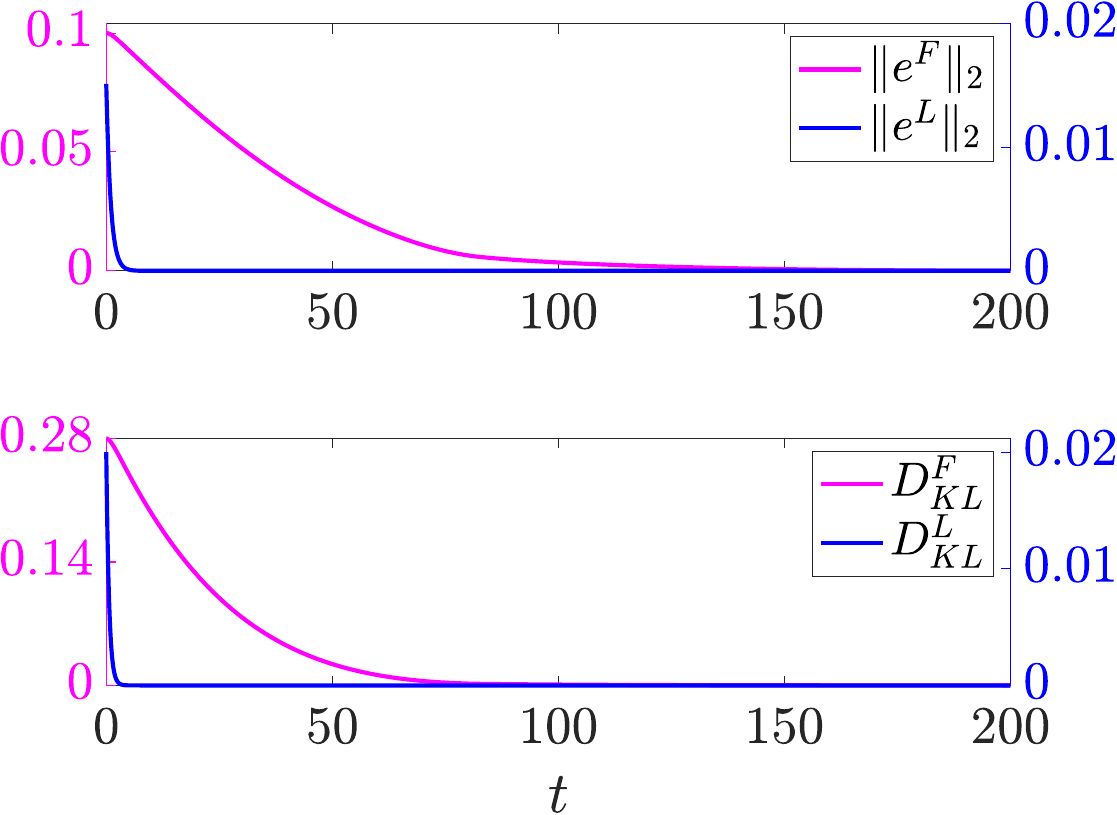}
        \caption{}
        \label{fig:errors_hd}
    \end{subfigure}
    %\hspace{0.1\textwidth}
    \caption{Validation of the higher dimensional strategy: (a) followers' density at the end of the trial, (b) leaders' density at the end of the trial, (c) time evolution of the percentage errors (top panel) and of the KL divergences (bottom panel).}
    \label{fig:validation_2d}
\end{figure*}

In this Section we validate the higher dimensional extension considering the two-dimensional domain $\Omega = [-\pi,\pi]^2$ and choosing as reference density the bimodal von Mises distribution
\begin{align} \label{eq:von_mises_hd}
    \bar{\rho}^F(\mathbf{x}) = Z \exp\{\boldsymbol{\kappa}^\top \mathbf{c}_1(\mathbf{x},\mu,\nu) + \mathbf{c}_2(\mathbf{x},\mu,\nu)\mathbf{I}_2 \mathbf{c}_2(\mathbf{x},\mu,\nu)^\top \},
\end{align}
where $\boldsymbol{\kappa}=[\kappa_1,\kappa_2]^\top$ is the vector of concentration coefficients, $\mu$,$\nu\in\Omega$ are the means along each directions, $\mathbf{c}_1=[\cos(x_1-\mu),\cos(x_2-\nu)]$, $\mathbf{c}_2=[\cos(x_1-\mu),\sin(x_2-\nu)]$, where $x_1$,$x_2$ are the components of $\mathbf{x}$ in the Cartesian coordinate system, $\mathbf{I}_2$ is the second order identity matrix and $Z$ is a normalization coefficient to allow $\bar{\rho}^F$ to sum to the total mass of followers. Also in this case, we approximate spatial derivatives with a central difference scheme over a $50\times50$ mesh-grid, and employ a forward Euler scheme for time integration, with time step $\Delta t=0.01$ over $200$ time units. We consider followers' reference density to have equal concentration coefficients and means along the directions, specifically $\kappa_1=\kappa_2 = 1$ and $\mu=\nu=0$. The follower-follower interactions are modeled considering the morse interaction kernel in \eqref{eq:ff_kernel}. Specifically, we recover the high-dimensional periodic kernel from its non-periodic counterpart 
\begin{align} \label{eq:kernel_hd}
    \hat{\mathbf{f}}^{FF}_i(\mathbf{x}) = \begin{cases}
        \frac{\mathbf{x}}{\|\mathbf{x}\|_2} \mathrm{e}^{-\frac{\mathbf{x}}{\ell_i}} \quad \mathrm{if} \,  &\|\mathbf{x}\|_2 \neq 0 \\
        0 \quad &\mathrm{otherwise}
    \end{cases} \, ,
\end{align}
where we set $\ell_r=\pi/2$, $\ell_a = \pi$ and $\zeta=1$; the length of the repulsive leader-follower interaction is set to $\ell = \pi$. We choose $D=0.01$, $M^F=0.4$ and $M^L=0.6$ to satisfy feasibility constraints. Results are shown in Fig. \ref{fig:validation_2d}, in which leaders and followers' densities start from uniform initial distributions and effectively converge to the reference.

To quantify performance, we consider density errors $\| e^i \|_2$, $i=L,F$, and KL divergences defined in \eqref{eq:KL_divergence}. Fig. \ref{fig:validation_2d} shows the simulation trial. Panels (a) and (b) show followers' and leaders' final densities respectively, while in panel (c) we report the time evolution of the percentage errors and the KL divergences.
}

\section{Deployment on Finite Multi-Agent Systems}
\label{sec:deployment}

Our framework exploits a macroscopic agents' description to perform control design at the density level.
To validate its applicability in real-world settings where the assumption of swarms of infinite size is unrealistic, we adopt a multi-scale control architecture grounded in the continuification paradigm \cite{maffettone2022continuification, nikitin2021continuation}: a ``macro-to-micro bridge'' discretizes the control field $\B{u}(\B{x},t)$ to individual leader commands via spatial collocation, while a ``micro-to-macro bridge'' reconstructs densities from agent positions through kernel estimation (see Fig. \ref{fig:multiscale} for a block diagram).

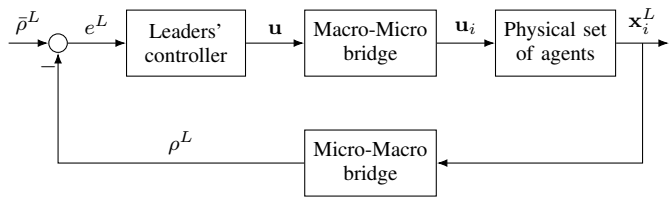
\begin{figure}
\centering
\begin{tikzpicture}[
    font=\footnotesize,
    x=1pt,
    y=1pt,
    block/.style={
        draw,
        rectangle,
        minimum width=45pt,
        minimum height=25pt,
        align=center
    },
    sum/.style={
        draw,
        circle,
        minimum size=7pt,
        inner sep=0pt
    },
    arrow/.style={
        -{Latex[length=4pt]},
        line width=0.4pt
    },
    node distance=22pt
]

% Nodes
\node[sum] (sum) {};
\node[block, right=of sum] (controller) {Leaders'\\controller};
\node[block, right=of controller] (discretization) {Macro-Micro\\bridge};
\node[block, right=of discretization] (agents) {Physical set\\of agents};
\node[block, below=20pt of discretization] (estimation) {Micro-Macro\\bridge};

% Input reference
\draw[arrow] ([xshift=-15pt]sum.west) -- (sum.west)
    node[midway, above] {$\bar{\rho}^L$};

% Sum sign
\node at ([yshift=-9pt]sum.west) {$-$};

% Forward path
\draw[arrow] (sum) -- (controller)
    node[midway, above] {$e^L$};

\draw[arrow] (controller) -- (discretization)
    node[midway, above] {$\mathbf{u}$};

\draw[arrow] (discretization) -- (agents)
    node[midway, above] {$\mathbf{u}_i$};

\draw[arrow] (agents) -- ([xshift=20pt]agents.east)
    node[midway, above] {$\mathbf{x}_i^L$};

% Output to estimation
\draw[arrow] ([xshift=10pt]agents.east) |- (estimation);

% Feedback
\draw[arrow] (estimation.west) -| (sum.south)
    node[pos=0.25, above] {$\rho^L$};

\end{tikzpicture}
\caption{Multiscale control scheme.}
\label{fig:multiscale}
\end{figure}

\subsection{Microscopic Model}

We consider the discrete counterpart of \eqref{eq:hd_model}, that is
\begin{subequations}
\label{eq:micro_model}
\begin{align}
\dot{\B{x}}_i^L &= \B{u}_i, \quad i = 1, \ldots, N^L \label{eq:micro_leaders} \\
\mathrm{d}\B{x}_k^F &= \frac{1}{N^L + N^F} \Bigg(\sum_{j=1}^{N^L} \B{f}^{FL}(\B{x}_j^L \triangleright \B{x}_k^F) \nonumber \\
&\quad\quad\quad + \sum_{m =1}^{N^F} \B{f}^{FF}(\B{x}_m^F \triangleright \B{x}_k^F) \Bigg)\mathrm{d}t + \sqrt{2D} \, \mathrm{d}\B{B}_k, \label{eq:micro_followers}
\end{align}
\end{subequations}
for $k = 1, \ldots, N^F$, where $\B{x}_i^L, \B{x}_k^F \in \Omega = [-\pi, \pi]^d$ (with $d=1, 2, 3$) denote leader and follower positions, $\B{u}_i$ is the control input for leader $i$, $\B{B}_k$ are independent standard Wiener processes, and $\B{f}^{FL}$, $\B{f}^{FF}$ are the follower-leader and follower-follower interaction kernels. 

\begin{rem}[Consistency with the macroscopic model]
The microscopic dynamics \eqref{eq:micro_model} represent the discrete counterpart of the macroscopic PDEs \eqref{eq:hd_model}. The leader-follower interaction term corresponds to the convolution $\B{v}^{FL} = \B{f}^{FL} * \rho^L$, while the follower-follower term corresponds to $\B{v}^{FF} = \B{f}^{FF} * \rho^F$. The case $d=1$ resembles the macroscopic model in \eqref{eq:themodel}.
\end{rem}

\subsection{One-dimensional Setting} 
We first consider the case $d=1$, for which we can proceed fully analytically as we are able to compute deconvolutions in closed form.

\subsubsection{Discrete Feasibility Analysis}
\label{subsec:discrete_feasibility}
The macroscopic feasibility conditions of Theorem~\ref{th:feasibility_theorem} establish bounds on the leader mass $M^L$ required to achieve a target distribution. In finite populations, these translate to constraints on the number of leaders $N^L$ for a given number of followers $N^F$. Specifically, recalling that $M^L = N^L/(N^L + N^F)$, that is the fraction of leaders in the discrete setting, the left constraint of Theorem~\ref{th:feasibility_theorem} becomes
\begin{equation}
\label{eq:discrete_feasibility}
N^L \geq \frac{\widehat{M}^L_1}{1 - \widehat{M}^L_1} N^F = \widehat{N}^L_1,
\end{equation}
providing an explicit lower bound on the leader count as a function of the follower population size. A similar bound $\widehat{N}^L_2$ can be recovered from $\widehat{M}^L_2$ -- see \eqref{eq:constraints_feasibility}.

\subsubsection{Macro-to-Micro and Micro-to-Macro Bridges} \label{sec:micro_control}
Given the macroscopic control field $u(x,t)$ from \eqref{eq:control_u}, we assign microscopic inputs via spatial collocation, i.e.~by setting:
\begin{equation}
\label{eq:discretization}
u_i(t) = u(x_i^L(t), t), \quad i = 1, \ldots, N^L.
\end{equation}
Such a procedure represents a discretization of the continuum control actions, i.e.~a macro-to-micro bridge (cf Fig.~\ref{fig:multiscale}).

A micro-to-macro bridge is developed to recover macroscopic densities from microscopic quantities. Specifically, the leader density $\rho^L$, required for evaluating $u$, is estimated online from agent positions using kernel density estimation with von Mises kernels.

\subsubsection{Microscopic Numerical Validation}

\begin{figure*}
    \centering
    \begin{subfigure}{0.32\textwidth}
        \centering
        \includegraphics[width=\textwidth]{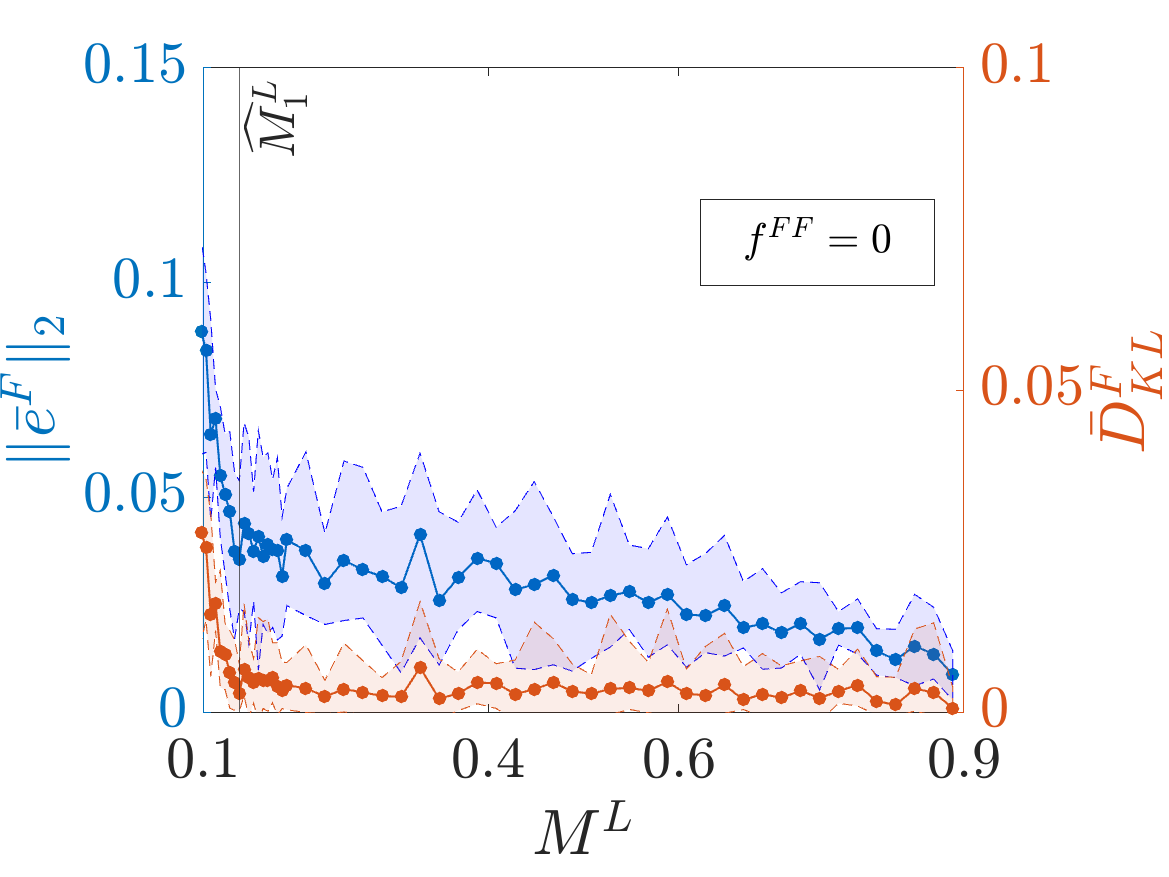}
        \caption{}
        \label{fig:E_DKL_no_interactions_def}
    \end{subfigure}
    \begin{subfigure}{0.32\textwidth}
        \centering
        \includegraphics[width=\textwidth]{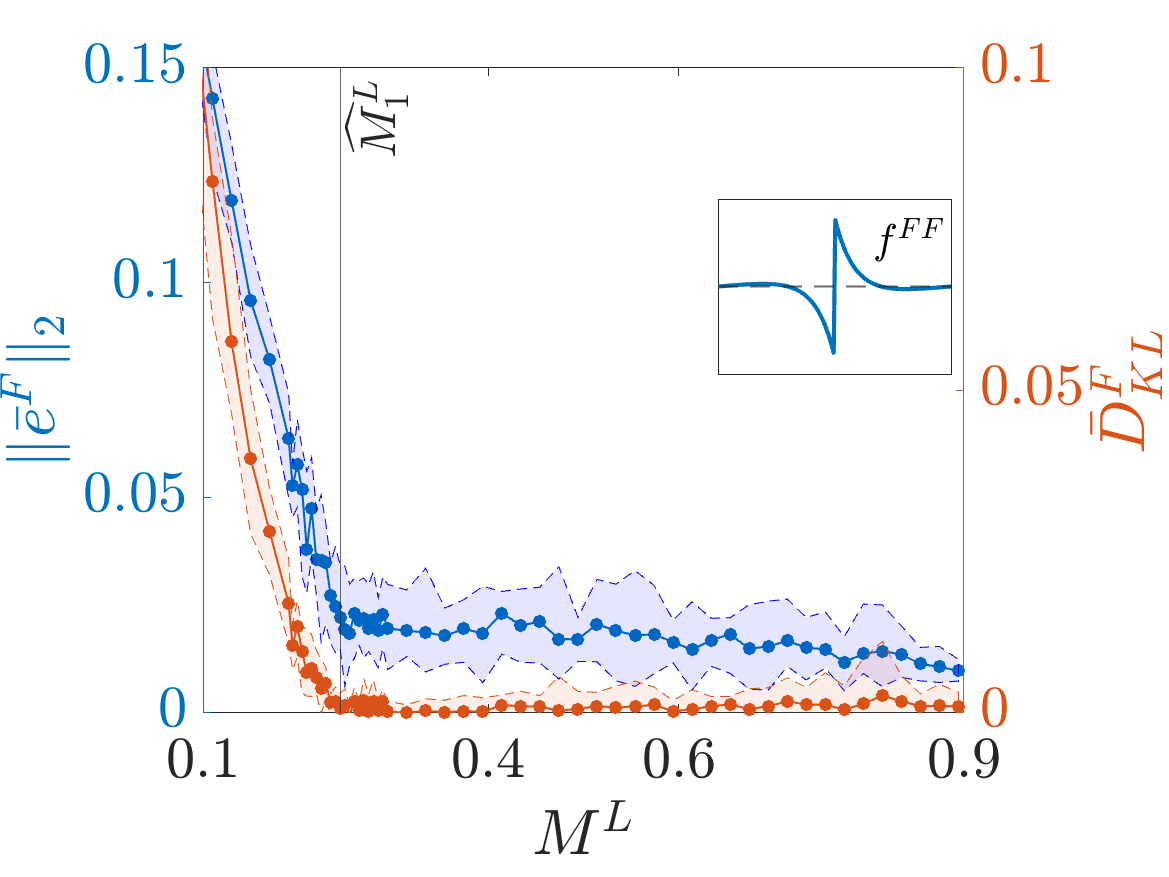}
        \caption{}
        \label{fig:E_DKL_weak_interactions_def}
    \end{subfigure}
    \begin{subfigure}{0.32\textwidth}
        \centering
        \includegraphics[width=\textwidth]{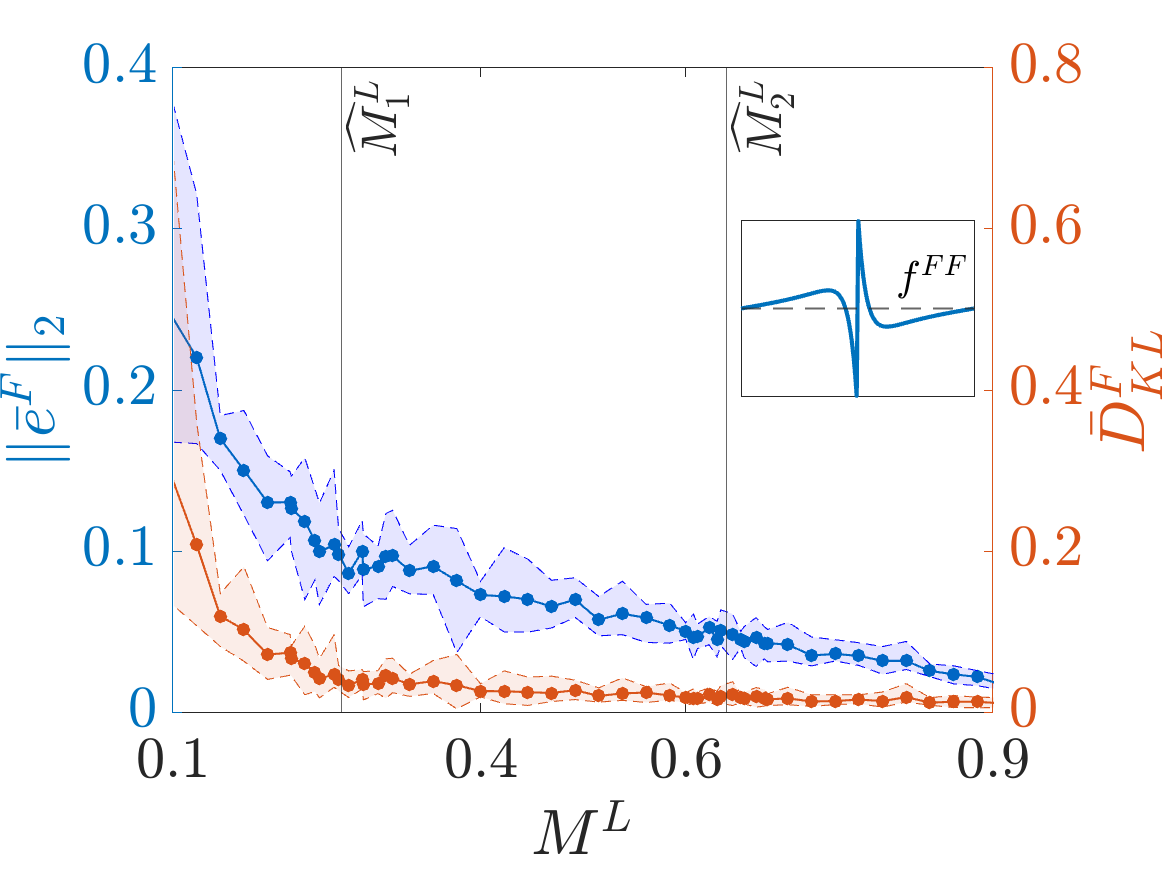}
        \caption{}
        \label{fig:E_DKL_strong_interactions_def}
    \end{subfigure}
    %\hspace{0.1\textwidth}
    \caption{Agent-based deployment: steady-state followers' error and KL divergence as a function of the leaders' mass in (a) absence, (b) weak and (c) strong follower-follower interactions. The shaded areas represent the minimum and maximum values.}
    \label{fig:discrete_feasibility}
\end{figure*}
We propose the agent-based version of the last trial discussed in Section~\ref{sec:numerical_validation} and showed in Fig.~\ref{fig:varying_ml}. Specifically, we conduct repeated simulations across different ranges of $N^L/N^F$ ratios. We keep the total number of agents constant to $N^L+N^F=500$ and vary the percentage of leaders within the range $M^L\in[0.1, 0.9]$. For each configuration we initialize equally spaced agents and control leaders according to \eqref{eq:discretization}. For each value of $M^L$, we  perform 20 simulations for $T=100$ time units with time step $\Delta t = 0.01$, according to the Euler scheme for the leaders and Euler-Maruyama scheme for the followers.
% For each configuration, we initialize $N^F$ followers uniformly in $\Omega$ and deploy $N^L$ leaders according to the discretized control law \eqref{eq:discretization}.
We measure the followers' density error $\| e^F \|_2$ and compute the average of its final value over the 20 simulations, referred to as $\| \bar{e}^F \|_2$ in the plots. Fig.~\ref{fig:discrete_feasibility} shows the results for the von Mises target distribution \eqref{eq:vonMises} with zero mean in the same three representative scenarios as in Section \ref{sec:numerical_validation}
: (a) no follower-follower interactions (other parameters: $\mu=0$, $\kappa = 1$, $D = 0.04$, $K=1$); (b) weak follower-follower Morse interactions, with $\ell_r = \pi/2$, $\ell_a = \pi$ and $\alpha = 1$ (other parameters: $\mu = 0$, $\kappa=1$, $D = 0.02$, $K=1$); (c) strong follower-follower Morse interactions, that is $\ell_r = \pi/15$, $\ell_a = \pi/2$ and $\alpha = 2$ (other parameters: $\mu = 0$, $\kappa = 2$, $D=0.16$,$K=1$). Also in this case we sample more values of $M^L$ close to the theoretical feasibility threshold. 

Comparing Fig.~\ref{fig:varying_ml} and Fig.~\ref{fig:discrete_feasibility}, we make the following observations. The macroscopic guarantees of Theorem~\ref{th:feasibility_theorem} and Theorem~\ref{thm:follower_stability} hold in the limit $N^L, N^F \to \infty$. For finite populations, a residual steady-state error persists due to (i) the finite-size approximation of the continuum and (ii) stochastic fluctuations in the follower dynamics; characterizing this error as a function of population size is discussed as an open problem in Section~\ref{sec:conclusions}. However, in the agent-based setting, the feasibility thresholds $\widehat{M}^L_{1,2}$, although computed within a continuum framework, still capture the regions where the steady-state error is minimized, highlighting the applicability of continuum approaches to real swarm robotics problems. Moreover, Fig.~\ref{fig:E_DKL_strong_interactions_def} shows improved steady-state performance above $\widehat{M}^L_2$. This does not contradict our theoretical findings, as the curves in Fig.~\ref{fig:E_DKL_strong_interactions_def} remain consistently above those in Fig.~\ref{fig:macro_ml_strong_ff}, indicating that finite-size and discretization errors prevent the observation of the performance degradation predicted above $\widehat{M}^L_2$ in the continuum.

The scaling \eqref{eq:discrete_feasibility} provides a quantitative design rule analogous to herdability conditions in shepherding problems (cf.\ Remark~\ref{rem:herdability}): for a given follower population and interaction parameters, the minimum number of leaders required for density control is explicitly determined by the continuum feasibility analysis.

\subsection{Two-dimensional setting}
We now consider $d=2$ and provide a description on how to deploy the macroscopic control strategy to $N$ agents in the periodic square of size $2 \pi$. After choosing the desired followers' density profile $\bar{\rho}^F$, we fix $\bar{\B{v}}^{FL}$ as in \eqref{eq:v_FL_hd} and, since a closed form of $\B{f}^{FL}$ was not found in higher dimensions, we recover $\bar{\rho}^L$ by performing numerical deconvolution as in \eqref{eq:deconv_hd} (see \cite{di2025continuification} for details on the numerical deconvolution). We fix $A=-\min_{\B{x}} \bar{\rho}^L(\B{x})$ in \eqref{eq:deconv_hd} to obtain the non-negative leaders' density with minimum integral and recover the minimum leaders' mass $\widehat{M}^L$ by computing such an integral. 

To validate this procedure, we consider the zero-mean Von-Mises distribution in \eqref{eq:von_mises_hd} with $\kappa_1=\kappa_2=1$, Morse-type follower-follower interactions with $\ell_r=\pi/2$, $\ell_a=\pi$ and $\zeta=1$, repulsive leader-follower interactions with $\ell = \pi$ and $D=0.01$. We deploy $N=1000$ agents, among which $N^F=660$ are followers and $N^L=340$ are leaders, to satisfy the feasibility constraints. We run the simulation over $200$ time units with time step $\Delta t = 0.01$ with the Euler-Maruyama scheme for the followers and Euler scheme for the leaders. Fig. \ref{fig:agent_based_validation2D} shows the initial and final configurations of agents and the evolution of percentage errors and KL divergences. In panel \ref{fig:errors_agents2D}, a small steady-state error is observed, unlike in the macroscopic simulations based on the continuum formulation presented in Section \ref{sec:macroscopic_validation_hd}. This discrepancy is mainly attributable to two factors: the finite size of the swarm, which limits the validity of the continuum assumption, and the stochastic behavior exhibited by the followers.

\begin{figure*}
    \centering
    \begin{subfigure}{0.3\textwidth}
        \centering
        \includegraphics[width=\textwidth]{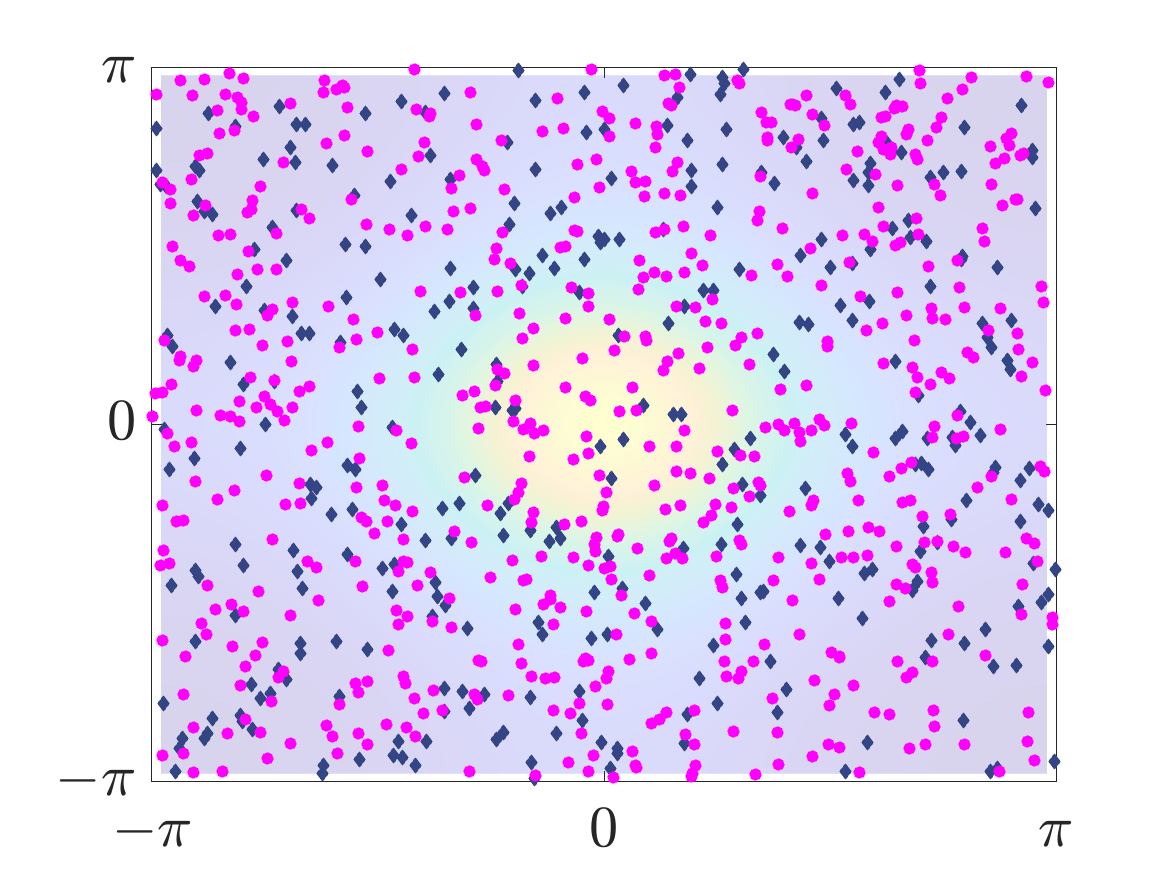}
        \caption{}
        \label{fig:initial_agents2D}
    \end{subfigure}
    % \hspace{0.1\textwidth}
    \hfill
    \begin{subfigure}{0.3\textwidth}
        \centering
        \includegraphics[width=\textwidth]{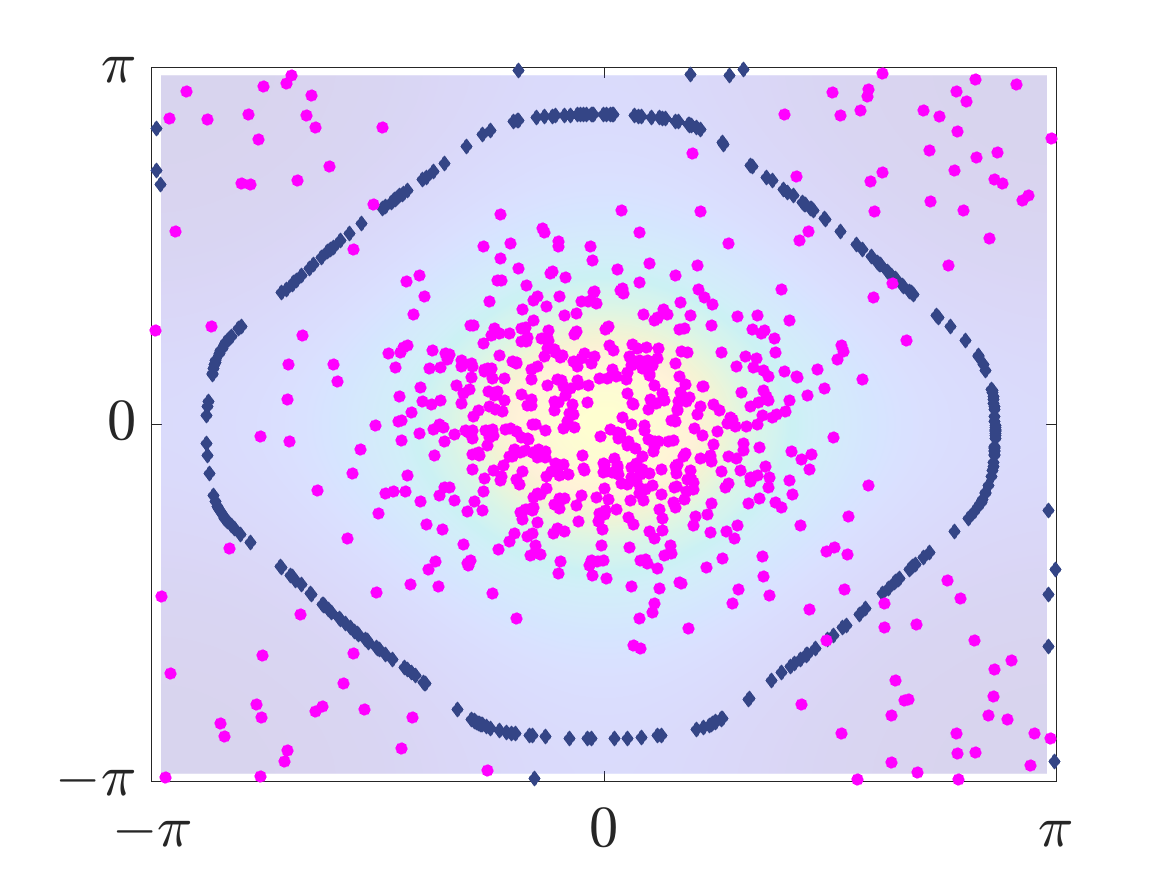}
        \caption{}
        \label{fig:final_agents2D}
    \end{subfigure}
    \hfill
    \begin{subfigure}{0.3\textwidth}
        \centering
        \includegraphics[width=\textwidth]{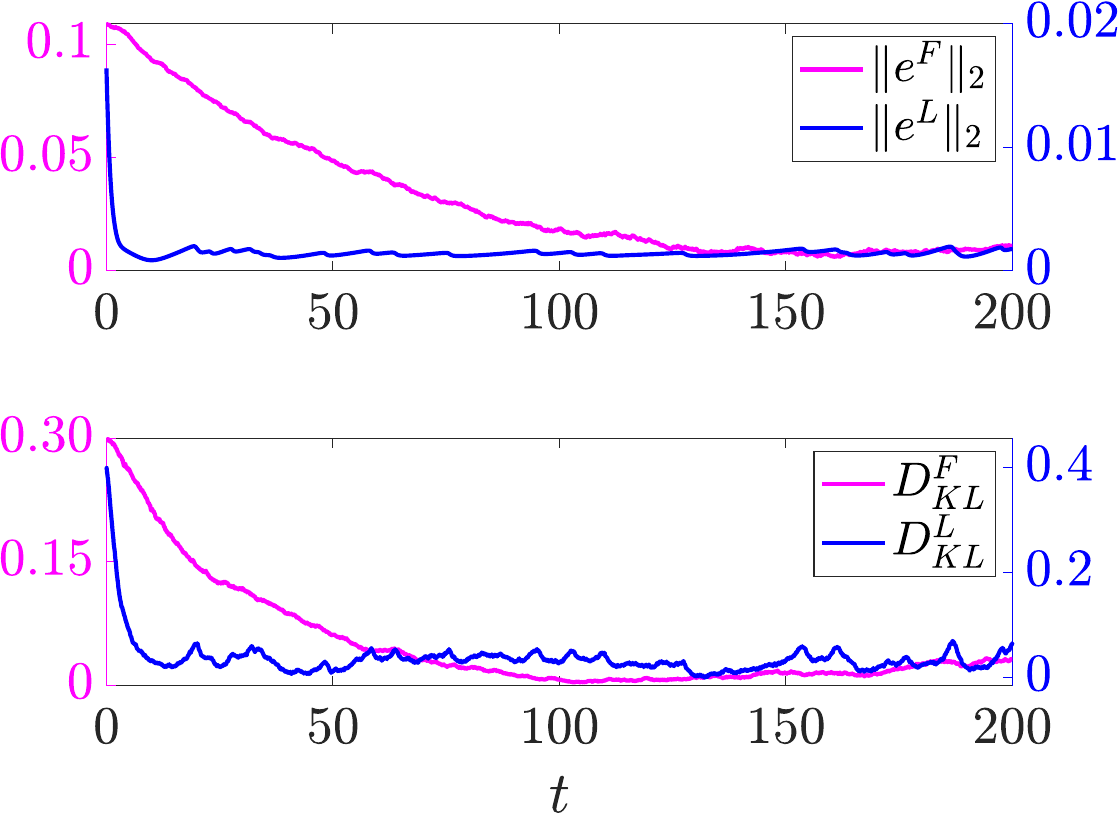}
        \caption{}
        \label{fig:errors_agents2D}
    \end{subfigure}
    %\hspace{0.1\textwidth}
    \caption{Agent-based deployment in a two-dimensional setting (blue diamonds are the leaders, magenta dots are the followers, the shaded background is $\bar{\rho}^F$): (a) initial configuration of agents; (b) final configuration of agents; (c) leaders and followers percentage errors and KL divergences.}
    \label{fig:agent_based_validation2D}
\end{figure*}

\section{Conclusions}\label{sec:conclusions}
In this work, we consider density control problems involving different populations of agents. Specifically, we consider an indirect control problem in which control authority is delegated only to population of leaders, which is tasked of driving a population of followers towards a desired distribution. Differently from our previous work \cite{maffettone2025leader}, we do not neglect inter-follower interactions, developing a generalized set-up in which we derive necessary and sufficient feasibility conditions and local stability of the control scheme.

A key finding is that follower interactions can either facilitate or make more cumbersome the density control task. Specifically, when the desired density is statistically close to the open-loop density displacement of followers, the task may require less leaders to be tackled. This is partly coherent with shepherding problems \cite{strombom2014solving}, where cohesion aids confinement tasks.

Several limitations of the present framework warrant discussion. First, the analysis assumes perfect knowledge of the density field $\rho^L(x,t)$. In practice, this must be estimated from the positions of a finite number of agents, introducing estimation errors whose effect on closed-loop stability is not captured by the current theory. Second, while Theorem~\ref{thm:follower_stability} and \ref{th:foll_stab_hd} guarantee asymptotic convergence, explicit finite-time bounds on the convergence rate are not provided; deriving such bounds, potentially via input-to-state stability arguments, remains an open problem. Third, the model assumes followers respond instantaneously to leaders through the interaction kernel; incorporating delays or second-order dynamics would require substantial extensions. %Finally, the feasibility and stability conditions depend on global properties of the desired target distribution (e.g., $\|g_1\|_\infty$), which may be difficult to verify for complex, multi-modal targets.

This work establishes the theoretical foundations for density control with interacting followers and lays the groundwork for several extensions currently under investigation. On the methodological side, we are exploring the integration of safety constraints through mean-field Control Barrier Functions \cite{fung2025mean, gao2026banach} to guarantee obstacle avoidance at the population level. From a theoretical perspective, the connection between feasibility conditions and herdability warrants further investigation. The minimum leader mass $\widehat{M}^L_1$ derived in Theorem~\ref{th:feasibility_theorem} provides a continuum analogue of herdability conditions in shepherding problems \cite{lama2024shepherding}. Extending this analysis to characterize \emph{herdability in the continuum}, determining which target distributions are reachable from a given initial configuration and with what minimum control authority, would bridge density control theory with the broader shepherding literature and provide principled design guidelines for multi-agent deployments. 

A further open problem concerns the quantitative characterization of finite-size effects. The  macroscopic stability guarantees hold exactly only in the continuum limit  $N^L, N^F \to \infty$; for finite populations, a residual steady-state  error persists due to discretization of the density field and stochastic fluctuations in the follower dynamics. Deriving explicit bounds on this error as a function of $N^L$, $N^F$, the diffusion coefficient $D$, and the interaction kernel parameters would provide rigorous finite-population design guidelines.

From a practical standpoint, experimental validation on robotic platforms 
is underway, leveraging the multi-scale control architecture described in 
Section~\ref{sec:deployment}: macroscopic density specifications are 
translated into individual agent commands via
discretization, and density feedback is reconstructed online from agent 
positions through kernel density estimation. 

Finally, the incorporation of learning-based methods can offer a complementary 
direction for settings where closed-form feasibility conditions are 
intractable; for instance, when interaction kernels $f^{FF}$ are unknown 
and must be identified from trajectory data, or when the target density 
$\bar{\rho}^F$ varies in real time and offline feasibility verification 
is impractical. In such cases, physics-informed neural networks or 
Gaussian process regression could be used to approximate the kernel 
and feasibility boundary online.

\section*{Acknowledgment}
The authors acknowledge Dr. Davide Salzano (University of Naples Federico II) for all the insightful discussion and for his help in developing Appendix A.

\section*{APPENDIX}
\subsection{Estimation of the basin of attraction}\label{app:basin_of_attraction}
The stable invariant manifold $W_s$ divides the state-space of \eqref{eq:aug_sys} in two regions, one of trajectories converging to the origin and one of trajectories such that $\eta\to\infty$ and $\xi\to0$. Hence, being the dynamics of $\xi$ linear and stable, a conservative estimate of the basin of attraction is given by the region below the nullcline $\eta_t = 0$, where $\eta_t$ switches sign. Such a nullcline takes the form
\begin{align}
    \xi = \frac{\alpha\eta-\delta\eta\sqrt{\eta}}{\beta\eta+c\sqrt{\eta}}.
\end{align}
Recalling the constraint $\xi(0) = 1$, we look for intersections of the nullcline with the horizontal axis $\xi = 1$. Doing so using the change of viariables $\tilde{\eta} = \sqrt{\eta}$, we find two possible intersections
\begin{align}
    \eta_{1,2} = \left(\frac{-(\beta-\alpha) \mp \sqrt{(\beta-\alpha)^2 - 4\gamma\delta}}{2\delta}\right)^2.
\end{align}
The greatest among the two intersections, say $\eta_2$, if exists, defines the basin of attraction as
\begin{align}
    \mathcal{B}(\mathbf{0}) = \{(\eta, \xi):\xi=1, \eta < \eta_2\}.
\end{align}
We refer to Fig. \ref{fig:ras_estimation} for a graphical interpretation, where we both consider the case $\eta_2$ exists or not.

%%%%%%%%%%%%%%%%%%%%%%%%%%%%%%%%%%%%%%%%%%%%%%%%%%%%%%%%%%%%%%%%%%%%%%%%%%%%%%%%
\section*{References}
\bibliographystyle{IEEEtran}

\begin{thebibliography}{10}
\providecommand{\url}[1]{#1}
\csname url@samestyle\endcsname
\providecommand{\newblock}{\relax}
\providecommand{\bibinfo}[2]{#2}
\providecommand{\BIBentrySTDinterwordspacing}{\spaceskip=0pt\relax}
\providecommand{\BIBentryALTinterwordstretchfactor}{4}
\providecommand{\BIBentryALTinterwordspacing}{\spaceskip=\fontdimen2\font plus
\BIBentryALTinterwordstretchfactor\fontdimen3\font minus
  \fontdimen4\font\relax}
\providecommand{\BIBforeignlanguage}[2]{{%
\expandafter\ifx\csname l@#1\endcsname\relax
\typeout{** WARNING: IEEEtran.bst: No hyphenation pattern has been}%
\typeout{** loaded for the language `#1'. Using the pattern for}%
\typeout{** the default language instead.}%
\else
\language=\csname l@#1\endcsname
\fi
#2}}
\providecommand{\BIBdecl}{\relax}
\BIBdecl

\bibitem{siri2021freeway}
S.~Siri, C.~Pasquale, S.~Sacone, and A.~Ferrara, ``Freeway traffic control: A
  survey,'' \emph{Automatica}, vol. 130, p. 109655, 2021.

\bibitem{papageorgiou2003review}
M.~Papageorgiou, C.~Diakaki, V.~Dinopoulou, A.~Kotsialos, and Y.~Wang, ``Review
  of road traffic control strategies,'' \emph{Proceedings of the IEEE},
  vol.~91, no.~12, pp. 2043--2067, 2003.

\bibitem{sinigaglia2025robust}
C.~Sinigaglia, A.~Manzoni, F.~Braghin, and S.~Berman, ``Robust optimal density
  control of robotic swarms,'' \emph{Automatica}, vol. 176, p. 112218, 2025.

\bibitem{dorigo2021swarm}
M.~Dorigo, G.~Theraulaz, and V.~Trianni, ``Swarm robotics: Past, present, and
  future [point of view],'' \emph{Proceedings of the IEEE}, vol. 109, no.~7,
  pp. 1152--1165, 2021.

\bibitem{brambilla2013swarm}
M.~Brambilla, E.~Ferrante, M.~Birattari, and M.~Dorigo, ``Swarm robotics: a
  review from the swarm engineering perspective,'' \emph{Swarm Intelligence},
  vol.~7, no.~1, pp. 1--41, 2013.

\bibitem{yuan2023multi}
Z.~Yuan, T.~Zheng, M.~Nayyar, A.~R. Wagner, H.~Lin, and M.~Zhu,
  ``Multi-robot-assisted human crowd control for emergency evacuation: A
  stabilization approach,'' in \emph{2023 American Control Conference (ACC)},
  2023.

\bibitem{helbing2000simulating}
D.~Helbing, I.~Farkas, and T.~Vicsek, ``Simulating dynamical features of escape
  panic,'' \emph{Nature}, vol. 407, no. 6803, pp. 487--490, 2000.

\bibitem{massana2022rectification}
H.~Massana-Cid, C.~Maggi, G.~Frangipane, and R.~Di~Leonardo, ``Rectification
  and confinement of photokinetic bacteria in an optical feedback loop,''
  \emph{Nature Communications}, vol.~13, no.~1, p. 2740, 2022.

\bibitem{giusti2026data}
A.~Giusti, D.~Salzano, M.~di~Bernardo, and T.~E. Gorochowski, ``Data-driven
  inference of digital twins for high-throughput phenotyping of motile and
  light-responsive microorganisms,'' \emph{Journal of The Royal Society
  Interface}, vol.~23, no. 234, 2026.

\bibitem{stern2018dissipation}
R.~E. Stern, S.~Cui, M.~L. Delle~Monache, R.~Bhadani, M.~Bunting, M.~Churchill,
  N.~Hamilton, R.~Haulcy, H.~Pohlmann, F.~Wu \emph{et~al.}, ``Dissipation of
  stop-and-go waves via control of autonomous vehicles: Field experiments,''
  \emph{Transportation research part C: emerging technologies}, vol.~89, pp.
  205--221, 2018.

\bibitem{pierson2017controlling}
A.~Pierson and M.~Schwager, ``Controlling noncooperative herds with robotic
  herders,'' \emph{IEEE Transactions on Robotics}, vol.~34, no.~2, pp.
  517--525, 2017.

\bibitem{almi2023optimal}
S.~Almi, M.~Morandotti, and F.~Solombrino, ``Optimal control problems in
  transport dynamics with additive noise,'' \emph{Journal of Differential
  Equations}, vol. 373, pp. 1--47, 2023.

\bibitem{maffettone2025leader}
G.~C. Maffettone, A.~Boldini, M.~Porfiri, and M.~d. Bernardo,
  ``Leader–follower density control of spatial dynamics in large-scale
  multiagent systems,'' \emph{IEEE Transactions on Automatic Control}, vol.~70,
  no.~10, pp. 6783--6798, 2025.

\bibitem{lama2024shepherding}
A.~Lama and M.~di~Bernardo, ``Shepherding and herdability in complex multiagent
  systems,'' \emph{Physical Review Research}, vol.~6, no.~3, p. L032012, 2024.

\bibitem{di2025continuification}
B.~Di~Lorenzo, G.~C. Maffettone, and M.~di~Bernardo, ``A continuification-based
  control solution for large-scale shepherding,'' \emph{European Journal of
  Control}, p. 101324, 2025.

\bibitem{strombom2014solving}
D.~Str{\"o}mbom, R.~P. Mann, A.~M. Wilson, S.~Hailes, A.~J. Morton, D.~J.
  Sumpter, and A.~J. King, ``Solving the shepherding problem: heuristics for
  herding autonomous, interacting agents,'' \emph{Journal of the royal society
  interface}, vol.~11, no. 100, p. 20140719, 2014.

\bibitem{napolitano2026optimal}
I.~Napolitano and M.~di~Bernardo, ``Optimal transport for time-varying
  multi-agent coverage control,'' \emph{arXiv preprint arXiv:2601.21753}, 2026.

\bibitem{maffettone2024mixed}
G.~C. Maffettone, L.~Liguori, E.~Palermo, M.~Di~Bernardo, and M.~Porfiri,
  ``Mixed reality environment and high-dimensional continuification control for
  swarm robotics,'' \emph{IEEE Transactions on Control Systems Technology},
  2024.

\bibitem{catello2025sparse}
L.~Catello, I.~Napolitano, D.~Salzano, and M.~di~Bernardo, ``Sparse shepherding
  control of large-scale multi-agent systems via reinforcement learning,''
  \emph{arXiv preprint arXiv:2511.21304}, 2025.

\bibitem{fornasier2014mean}
M.~Fornasier and F.~Solombrino, ``Mean-field optimal control,'' \emph{ESAIM:
  Control, Optimisation and Calculus of Variations}, vol.~20, no.~4, pp.
  1123--1152, 2014.

\bibitem{fornasier2014mean2}
M.~Fornasier, B.~Piccoli, N.~P. Duteil, and F.~Rossi, ``Mean-field optimal
  control by leaders,'' \emph{Proceedings of the 53rd IEEE Conference on
  Decision and Control}, pp. 6957--6962, 2014.

\bibitem{bongini2017optimal}
M.~Bongini and G.~Buttazzo, ``Optimal control problems in transport dynamics,''
  \emph{Mathematical Models and Methods in Applied Sciences}, vol.~27, no.~03,
  pp. 427--451, 2017.

\bibitem{albi2022mean}
G.~Albi, S.~Almi, M.~Morandotti, and F.~Solombrino, ``Mean-field selective
  optimal control via transient leadership,'' \emph{Applied Mathematics \&
  Optimization}, vol.~85, no.~2, p.~22, 2022.

\bibitem{albi2024kinetic}
G.~Albi and F.~Ferrarese, ``Kinetic description of swarming dynamics with
  topological interaction and transient leaders,'' \emph{Multiscale Modeling \&
  Simulation}, vol.~22, no.~3, pp. 1169--1195, 2024.

\bibitem{lama2025nonreciprocal}
A.~Lama, M.~di~Bernardo, and S.~H. Klapp, ``Nonreciprocal field theory for
  decision-making in multi-agent control systems,'' \emph{Nature
  Communications}, vol.~16, no.~1, p. 8450, 2025.

\bibitem{bernardi2021leadership}
S.~Bernardi, R.~Eftimie, and K.~J. Painter, ``Leadership {{Through Influence}}:
  {{What Mechanisms Allow Leaders}} to {{Steer}} a {{Swarm}}?'' \emph{Bulletin
  of Mathematical Biology}, vol.~83, no.~6, p.~69, Jun. 2021.

\bibitem{bernardi2019macroscopic}
\BIBentryALTinterwordspacing
S.~Bernardi, G.~Estrada-Rodriguez, H.~Gimperlein, and K.~J. Painter,
  ``Macroscopic descriptions of follower-leader systems,'' \emph{Kinetic and
  Related Models}, vol.~14, no.~6, pp. 981--1002, 2021. [Online]. Available:
  \url{https://www.aimsciences.org/article/id/d8004c4d-a77c-48ce-b44c-66d8d44c2c53}
\BIBentrySTDinterwordspacing

\bibitem{bernoff2011primer}
A.~J. Bernoff and C.~M. Topaz, ``A primer of swarm equilibria,'' \emph{SIAM
  Journal on Applied Dynamical Systems}, vol.~10, no.~1, pp. 212--250, 2011.

\bibitem{bodnar2005derivation}
M.~Bodnar and J.~J.~L. Velazquez, ``Derivation of macroscopic equations for
  individual cell-based models: a formal approach,'' \emph{Mathematical methods
  in the applied sciences}, vol.~28, no.~15, pp. 1757--1779, 2005.

\bibitem{boldini2024stigmergy}
A.~Boldini, M.~Civitella, and M.~Porfiri, ``Stigmergy: from mathematical
  modelling to control,'' \emph{Royal Society Open Science}, vol.~11, no.~9, p.
  240845, 2024.

\bibitem{royden1988real}
H.~L. Royden and P.~Fitzpatrick, \emph{Real analysis}.\hskip 1em plus 0.5em
  minus 0.4em\relax Macmillan New York, 1988, vol.~32.

\bibitem{kullback1951information}
S.~Kullback and R.~A. Leibler, ``On information and sufficiency,'' \emph{The
  annals of mathematical statistics}, vol.~22, no.~1, pp. 79--86, 1951.

\bibitem{khalil2002nonlinear}
H.~K. Khalil, \emph{Nonlinear systems}.\hskip 1em plus 0.5em minus 0.4em\relax
  Prentice Hall, 2002.

\bibitem{dilorenzo2025decentralized}
B.~Di~Lorenzo, G.~C. Maffettone, and M.~di~Bernardo, ``Decentralized
  continuification control of multi-agent systems via distributed density
  estimation,'' \emph{IEEE Control Systems Letters}, vol.~9, pp. 1580--1585,
  2025.

\bibitem{maffettone2022continuification}
G.~C. Maffettone, A.~Boldini, M.~Di~Bernardo, and M.~Porfiri,
  ``Continuification control of large-scale multiagent systems in a ring,''
  \emph{IEEE Control Systems Letters}, vol.~7, pp. 841--846, 2022.

\bibitem{nikitin2021continuation}
D.~Nikitin, C.~Canudas-de Wit, and P.~Frasca, ``A continuation method for
  large-scale modeling and control: from odes to pde, a round trip,''
  \emph{IEEE Transactions on Automatic Control}, vol.~67, no.~10, pp.
  5118--5133, 2021.

\bibitem{fung2025mean}
S.~W. Fung and L.~Nurbekyan, ``Mean-field control barrier functions: A
  framework for real-time swarm control,'' \emph{Proceedings of the 2025
  American Control Conference (ACC)}, 2025.

\bibitem{gao2026banach}
X.~Gao, G.~Pascual, S.~Brown, and S.~Mart{\'\i}nez, ``Banach control barrier
  functions for large-scale swarm control,'' \emph{arXiv preprint
  arXiv:2602.05011}, 2026.

\end{thebibliography}

\end{document}